\documentclass[graybox]{svmult}

\usepackage{type1cm}        
\usepackage{makeidx}         
\usepackage{graphicx}        
\usepackage{multicol}        
\usepackage[bottom]{footmisc}

\usepackage{newtxtext}       %
\usepackage{newtxmath}       

\usepackage{url}
\usepackage{amsmath}
\usepackage{amsfonts}
\usepackage{algorithm,algorithmic}
\usepackage{mathtools}
\usepackage{bm}
\usepackage{tikz}
\usetikzlibrary{calc}
\usepackage[all]{xy}
\usetikzlibrary{matrix,arrows,shapes}
\makeindex             

\begin{document}

\title*{Sampling and Statistical Physics via Symmetry}
\author{Steve Huntsman}
\institute{\email{sch213@nyu.edu}}
%
%
\maketitle

\abstract*{
We formulate both Markov chain Monte Carlo (MCMC) sampling algorithms and basic statistical physics in terms of elementary symmetries. This perspective on sampling yields derivations of well-known MCMC algorithms and a new parallel algorithm that appears to converge more quickly than current state of the art methods. The symmetry perspective also yields a parsimonious framework for statistical physics and a practical approach to constructing meaningful notions of effective temperature and energy directly from time series data. We apply these latter ideas to Anosov systems.
}

\abstract{
We formulate both Markov chain Monte Carlo (MCMC) sampling algorithms and basic statistical physics in terms of elementary symmetries. This perspective on sampling yields derivations of well-known MCMC algorithms and a new parallel algorithm that appears to converge more quickly than current state of the art methods. The symmetry perspective also yields a parsimonious framework for statistical physics and a practical approach to constructing meaningful notions of effective temperature and energy directly from time series data. We apply these latter ideas to Anosov systems.
}

\section{Introduction}
\label{sec:introduction}

Sampling and statistical physics are essentially dual concepts. Phenomenologists sample from physical models to obtain data, and theorists construct physical models to explain data. For simple data and/or systems, the pushforward of an initial state distribution under a deterministic dynamical model may be theoretically adequate (at least up to a Lyapunov time or its ilk), but for complex data and/or systems, an intrinsically statistical model is typically necessary.

Moreover, sampling strategies and physical models are frequently manifestations of each other \cite{amey2018analysis}. For instance, Glauber (spin-flip), Kawasaki (spin-exchange), and Swendsen-Wang (spin-cluster) dynamics are each both special-purpose \emph{Markov chain Monte Carlo} (MCMC) algorithms and models for the time evolution of a spin system. Each algorithm/model has its own physical features, e.g. spin-flip dynamics are suited to the canonical ensemble; spin-exchange dynamics preserve an order parameter; and spin-cluster dynamics are both more efficient and descriptive for systems near criticality \cite{binder2019monte}. As the \emph{chaotic hypothesis} essentially stipulates that spin systems are generic statistical-physical systems \cite{gallavotti1995stationary,gallavotti1999statistical}, this blurring of the distinction between algorithm and model can also be regarded as generic.
\footnote{
NB. Even $SU(N)$ field theory can be treated as a spin system: see, e.g. \cite{durhuus1980connection}.
}

Meanwhile, physics has a long tradition of formulating theories in terms of symmetries. Perhaps surprisingly, both sampling strategies and the basic structure of statistical physics itself can also be formulated in terms of symmetries. We outline these respective formulations with an eye towards (in \S \ref{sec:1}) efficient parallel MCMC algorithms and (in \S \ref{sec:2}) effective temperatures and energy functions that can be obtained directly from data for descriptive purposes. Finally, in \S \ref{sec:Anosov} we apply the ideas of \S \ref{sec:2} to Anosov systems, where they suggest a broader framework for nonequilibrium statistical physics.

\section{Sampling via symmetry}
\label{sec:1}

MCMC algorithms estimate expected values by running a Markov chain with the desired invariant measure. Though they arose from computational physics, MCMC algorithms have become ubiquitous, particularly in statistical inference and machine learning, and their importance in the toolkit of numerical algorithms is difficult to overstate \cite{richey2010mcmc,brooks2011mcmc}.

As such, there is a vast literature on MCMC algorithms. However, there is also much still left unexplored. As we shall see, the interface between MCMC algorithms and the theory of Lie groups and Lie algebras holds a surprise. A key observation is that the space of transition matrices with a given invariant measure is a \emph{monoid} that is closely related to a Lie group. Certain natural elements of this monoid with simple closed form expressions naturally lead to constructions of the classical Barker and Metropolis MCMC samplers. These constructions generalize, leading to higher-order versions of samplers that respectively correspond to the ensemble MCMC algorithm of \cite{neal2011ensemble} and an algorithm of \cite{delmas2009does}. A further generalization leads to a new algorithm that we call the \emph{higher-order programming solver} and whose convergence appears to improve on the state of the art. Each of these algorithms is only presently defined for finite state spaces and leaves the proposal mechanism unspecified: indeed, our entire focus is on acceptance mechanisms. 
\footnote{
By repeated sampling, we can extend any proposal mechanism for single states to multiple states.
}

In this section, which is based on the conference paper \cite{huntsman2020fast}, we review the basics of MCMC, Lie theory, and related work in \S \ref{sec:Background}. We then briefly consider the Lie group generated by a probability measure in \S \ref{sec:Group}. In particular, we construct a convenient basis of the subalgebra of the stochastic Lie algebra that annihilates a target probability measure $p$. This basis only requires knowledge of $p$ up to a multiplicative factor (e.g., a partition function), and this fact is the essential reason why MCMC algorithms work in general. In \S \ref{sec:Monoid}, we show how we can analytically produce transition matrices that leave $p$ invariant. We then construct the Barker and Metropolis samplers from Lie-theoretic considerations in \S \ref{sec:BarkerMetropolis}. In \S \ref{sec:Algebra}, we extend earlier results, leading to generalizations of the Barker and Metropolis samplers that entertain multiple proposals at once and that we explicitly construct in \S \ref{sec:HigherOrder}. We then demonstrate the behavior of these samplers on a small spin glass in \S \ref{sec:Behavior}. In \S \ref{sec:LP}, we outline the construction of multiple-proposal transition matrices that are closest in Frobenius norm to the ``ideal'' transition matrix $1p$, and we introduce and demonstrate the resulting higher-order programming solver. Finally, we close our discussion of MCMC algorithms with remarks in \S \ref{sec:RemarksPhysics}.

\subsection{\label{sec:Background}Background}

\subsubsection{\label{sec:MCMC}Markov chain Monte Carlo}

As we have already mentioned in \S \ref{sec:1} and (e.g.) \cite{bremaud1999markov} discusses at length, MCMC algorithms estimate expected values of functions with respect to a probability measure $p$ that is infeasible to construct. The archetypal instance comes from equilibrium statistical physics, where $p_j = Z^{-1} \exp(-\beta E_j)$ is hard to compute because the partition function $Z$ is unknown due to the scale of the problem, but the energies $E_j$ are individually easy to compute. The miracle of MCMC is that we can construct an irreducible, ergodic Markov chain with invariant measure $p$ using only unnormalized and easily computable terms such as $\exp(-\beta E_j)$.

Let $X_t$ denote the state of such a chain at time $t$. In the limit, $X_t \sim p$ for any initial condition, and $\mathbb{E}_p f(X) = \lim_{t \rightarrow \infty} \frac{1}{t} \sum_{j=1}^t f(X_j)$ even though the $X_j$ are correlated. The problem of constructing such a chain is typically decomposed into \emph{proposal} and \emph{acceptance} steps as in Algorithm \ref{alg:mcmc}, with respective probabilities $q_{jk} := \mathbb{P}(X' = k | X_t = j)$ and $\alpha_{jk} := \mathbb{P}(X_{t+1} = k | X' = k, X_t = j)$. The proposal and acceptance are combined to form the chain transitions via $P_{jk} := \mathbb{P}(X_{t+1} = k | X_t = j) = q_{jk} \alpha_{jk}$.

\begin{algorithm}[tb]
   \caption{MCMC}
   \label{alg:mcmc}
\begin{algorithmic}
   \STATE {\bfseries Input:} Runtime $T$ and $P_{jk} = q_{jk}\alpha_{jk}$ with $pP = p$
   \STATE Initialize $t=0$ and $X_0$
   \REPEAT
   \FOR{each state $k$}
   \STATE Propose $k$ with probability $q_{jk}$
   \ENDFOR
   \STATE Accept $X_{t+1} = k$ with probability $\alpha_{jk}$
   \STATE Set $t = t+1$
   \UNTIL{$t = T$}
   \STATE {\bfseries Output:} $\{X_t\}_{t=0}^T \sim p^{\times (T+1)}$ (approximately)
\end{algorithmic}
\end{algorithm}

The reasonably generic Hastings algorithm employs an acceptance mechanism of the form $\alpha_{jk} = \frac{s_{jk}}{1+t_{jk}}$, where $t_{jk} := \frac{p_j q_{jk}}{p_k q_{kj}}$ and $s$ need only be symmetric with entries $s_{jk} \in (0, 1 + \min(t_{jk},t_{kj})]$. The Barker sampler corresponds to the choice $s_{jk} = 1$, while the Metropolis-Hastings sampler corresponds to the optimal \cite{peskun1973optimum} choice $s_{jk} = 1 + \min(t_{jk},t_{kj})$.

\subsubsection{\label{sec:Lie}Lie groups and Lie algebras}

For the sake of self-containment, we briefly restate the basic concepts of Lie theory in the real and finite-dimensional setting.
For general background on Lie groups and algebras, see, e.g. \cite{onishchik1990lie,kirillov2008lie}. 

A \emph{Lie group} is a manifold with a smooth group structure. The tangent space of a Lie group $G$ at the identity is the \emph{Lie algebra} $\mathfrak{lie}(G)$: the group structure is echoed in the algebra via a bilinear antisymmetric \emph{bracket} $[\cdot, \cdot]$ that satisfies the \emph{Jacobi identity}
\begin{equation}
[X,[Y,Z]] + [Y,[Z,X]] + [Z,[X,Y]] = 0. \nonumber
\end{equation}

By Ado's theorem, a real finite-dimensional Lie group $G$ is isomorphic to a subgroup of the group $GL(n, \mathbb{R})$ of invertible $n \times n$ matrices over $\mathbb{R}$. In this circumstance, $\mathfrak{lie}(G)$ is isomorphic to a subalgebra of real $n \times n$ matrices, with bracket as the usual matrix commutator $[X,Y] := XY-YX$. Meanwhile, the matrix exponential sends $\mathfrak{lie}(G)$ to $G$ in a way that respects both the algebra and group structures.

\subsubsection{\label{sec:Related}Related work}

The higher-order Barker and Metropolis samplers we construct have previously been considered in \cite{neal2011ensemble} and \cite{delmas2009does}, respectively. Besides ensemble algorithms, \cite{robert2018mcmc} details approaches to accelerating MCMC algorithms via multiple try algorithms as in \cite{liu2000multiple,martino2018MT,martino2018IRSM}; and by parallelization as in \cite{calderhead2014parallel}. 

Discrete symmetries that (possibly approximately) preserve the level sets of a target measure have also been exploited to accelerate MCMC algorithms in \cite{niepert2012markov,niepert2012mcmc,bui2013automorphism,shariff2015symmetries,vandenbroeck2015lifted,anand2016contextual}. Similarly, ``group moves'' for MCMC algorithms were considered in \cite{liu1999parameter,liu2000generalised}. However, we are not aware of previous attempts to consider continuous symmetries preserving a target measure in the context of MCMC.

That said, Markov models on groups have been studied in, e.g., \cite{saloffcoste2001random,ceccherini2008harmonic}. However, although notional applications of Lie theory to Markov models motivate work on the stochastic group, actual applications themselves are few in number, with \cite{sumner2012lie} serving as an exemplar. 

If we ignore considerations of analytical tractability and/or computational efficiency, we can consider generic MCMC algorithms that optimize some criterion over the relevant monoid. Optimal control considerations lead to algorithms such as those of \cite{suwa2010detailed,chen2013accelerating,bierkens2016nonreversible,takahashi2016detailed} that optimize convergence while sacrificing reversibility/detailed balance. Meanwhile, \cite{frigessi1992optimal,pollet2004optimal,chen2012optimal,wu2015optimal,huang2018optimal} seek to optimize the asymptotic variance.

\subsection{\label{sec:Group}The Lie group generated by a probability measure}

For $1 < n \in \mathbb{N}$, let $p$ be a probability measure on $[n] := \{1,\dots,n\}$. Relying on context to resolve any ambiguity, we write $1 = (1,\dots,1)^T \in \mathbb{R}^n$. Now following \cite{johnson1985markov,poole1995stochastic,boukas2015lie,guerra2018stochastic}, we define the \emph{stochastic group} 
\begin{equation}
\label{eq:STO}
STO(n) := \{P \in GL(n,\mathbb{R}): P1 = 1 \},
\end{equation}
i.e., the stabilizer fixing $1$ on the left in $GL(n,\mathbb{R})$, and
\begin{equation}
\label{eq:GEN}
\langle p \rangle := \{P \in STO(n) : pP = p \},
\end{equation}
i.e., the stabilizer fixing $p$ on the right in $STO(n)$. We call $\langle p \rangle$ the \emph{group generated by $p$}. $STO(n)$ and $\langle p \rangle$ are both Lie groups, with respective dimensions $n(n-1)$ and $(n-1)^2$. If $P \in STO(n)$ is irreducible and ergodic, then its unique invariant measure is $\langle P \rangle := 1^T(P-I+11^T)^{-1}$. Now $pP = p$ iff $\langle p \rangle = \langle \langle P \rangle \rangle$, and $\langle p \rangle - I \subset \mathfrak{lie} ( \langle p \rangle ) \subset \mathfrak{lie}(STO(n))$.

For $(j,k) \in [n] \times [n-1]$, write 
\begin{equation}
\label{eq:STObasis}
e_{(j,k)} := e_j(e_k^T-e_n^T), 
\end{equation}
where $\{e_j\}_{j \in [n]}$ is the standard basis of $\mathbb{R}^n$. Now the matrices $\{e_{(j,k)}\}_{(j,k) \in [n] \times [n-1]}$ form a basis of $\mathfrak{lie}(STO(n))$ and 
\begin{align}
\label{eq:STOproduct}
e_{(j,k)} e_{(\ell, m)} & = e_j(e_k^T-e_n^T) e_\ell(e_m^T-e_n^T) \nonumber \\
& = (\delta_{k\ell}-\delta_{n\ell}) e_{(j,m)},
\end{align}
so upon considering $j \leftrightarrow \ell, k \leftrightarrow m$ we have that
\begin{equation}
\label{eq:STOcommutator}
\left[ e_{(j,k)}, e_{(\ell,m)} \right] = (\delta_{k\ell} - \delta_{n\ell})e_{(j,m)} - (\delta_{mj} - \delta_{nj})e_{(\ell,k)}.
\end{equation}
This basis has the obvious advantage of computationally trivial decompositions.


For $j,k \in [n-1]$, we set $r_j := p_j/p_n$ and
\begin{align}
\label{eq:GENbasis}
e_{(j,k)}^{(p)} := & \ e_{(j,k)} - r_j e_{(n,k)} \nonumber \\ = & \left ( e_j - r_j e_n \right ) (e_k^T-e_n^T).
\end{align}

\begin{svgraybox}
If $p_j \equiv \mathcal{L}_j/Z$, say with $\mathcal{L}_j \equiv \exp(-\beta E_j)$, then $r_j = \mathcal{L}_j/\mathcal{L}_n$ does not depend on $Z$ at all. This is the basic reason why MCMC methods can avoid grappling with normalization factors such as partition functions, and in turn why MCMC methods are so useful. 
\end{svgraybox}

For future reference, write $r := (r_1,\dots,r_{n-1},1)$ and $r^- := (r_1,\dots,r_{n-1})$.

\begin{lemma}
\label{lem:power}
For $i \in \mathbb{N}$,
\begin{equation}
\left (e_{(j,k)}^{(p)} \right )^i = \begin{cases} I, & i = 0; \\ \left ( \delta_{jk} + r_j \right )^{i-1} e_{(j,k)}^{(p)}, & i > 0. \end{cases}
\end{equation}
\end{lemma}

\begin{proof}
Using the rightmost expression in \eqref{eq:GENbasis} and using $j,k,\ell,m \ne n$ to simplify the product of the innermost two factors, we obtain  
\begin{equation}
\label{eq:GENproduct}
e_{(j,k)}^{(p)} e_{(\ell,m)}^{(p)} = \left ( \delta_{k \ell} + r_\ell \right ) e_{(j,m)}^{(p)}.
\end{equation}
Taking $j = \ell$ and $k = m$ establishes the result for $i \le 2$. The general case follows by induction on $i$. 
\end{proof}

\begin{theorem}
\label{thm:GENbasis}
The $e_{(j,k)}^{(p)}$ form a basis for $\mathfrak{lie} ( \langle p \rangle )$ and
\begin{equation}
\label{eq:GENcommutator}
\left[ e_{(j,k)}^{(p)}, e_{(\ell,m)}^{(p)} \right] = \left ( \delta_{k\ell} + r_\ell \right ) e_{(j,m)}^{(p)} - \left ( \delta_{mj} + r_j \right ) e_{(\ell,k)}^{(p)}.
\end{equation}
\end{theorem}

\begin{proof}
Note that
$p e_{(j,k)}^{(p)} = \left ( p_j - r_j p_n \right ) \left ( e_k^T-e_n^T \right ) \equiv 0$. 
Furthermore, linear independence and the commutation relations are both obvious, so we need only show that $\exp t e_{(j,k)}^{(p)} \in \langle p \rangle$ for all $t \in \mathbb{R}$. By Lemma \ref{lem:power}, 
\begin{align}
\label{eq:GENexp}
\exp t e_{(j,k)}^{(p)} & = I + e_{(j,k)}^{(p)} \sum_{i = 1}^\infty \frac{t^i \left ( \delta_{jk} + r_j \right )^{i-1}}{i!} \nonumber \\
& = I + \frac{e^{t(\delta_{jk} + r_j)} - 1}{\delta_{jk} + r_j} e_{(j,k)}^{(p)}. 
\end{align}
\end{proof}

For later convenience, we write 
\begin{equation}
\label{eq:fjkp}
f_{(j,k)}^{(p)}(t) := \frac{e^{-t(\delta_{jk} + r_j)} - 1}{\delta_{jk} + r_j}.
\end{equation}

\subsection{\label{sec:Monoid}The positive monoid of a measure}


Elements of $STO(n)$ need not be \emph{bona fide} stochastic matrices because they can have negative entries; on the other hand, stochastic matrices need not be invertible. We therefore consider the \emph{monoids} (i.e., semigroups with identity; compare \cite{hilgert1993lie})
\begin{equation}
STO^+(n) := \{P \in M(n,\mathbb{R}): P1 = 1 \text{ and } P \ge 0 \},
\end{equation}
where $P \ge 0$ is interpreted per entry, and
\begin{equation}
\langle p \rangle^+ := \{P \in STO^+(n) : pP = p \}.
\end{equation}
Note that $STO^+(n) \not \subset STO(n)$ and $\langle p \rangle^+ \not \subset \langle p \rangle$, since the left hand sides contain noninvertible elements. Also, $STO^+(n)$ and $\langle p \rangle^+$ are bounded convex polytopes.

\begin{lemma}
\label{lem:pos}
If $t_j \ge 0$, then $\exp \left ( -\sum_j t_j e_{(j,j)}^{(p)} \right ) \in \langle p \rangle^+$.
\end{lemma}

\begin{proof}
By hypothesis and \eqref{eq:GENbasis}, $-\sum_j t_j e_{(j,j)}^{(p)}$ has nonpositive diagonal entries and nonnegative off-diagonal entries; the result follows by regarding the sum as the generator matrix of a continuous-time Markov process.
\end{proof}

In particular, for $t \ge 0$ we have that 
\begin{equation}
\label{eq:explicitsemi1}
\exp \left (-t e_{(j,j)}^{(p)} \right ) = I + f_{(j,j)}^{(p)}(t) \cdot e_{(j,j)}^{(p)} \in \langle p \rangle^+,
\end{equation}
where $f_{(j,j)}^{(p)}$ is as in \eqref{eq:fjkp}. Unfortunately, aside from \eqref{eq:explicitsemi1}, Lemma \ref{lem:pos} does not give a convenient way to construct explicit elements of $\langle p \rangle^+$ in closed form. This situation is an analogue of the highly nontrivial quantum compilation problem (see \cite{dawson2006solovaykitaev}). 

Indeed, even if the sum in the lemma's statement has only two terms, we are immediately confronted with the formidable Zassenhaus formula (see \cite{casas2012zassenhaus}):
\begin{equation}
\label{eq:zassenhaus}
\exp(t(X+Y)) = \exp(tX) \exp(tY) \prod_{i = 2}^\infty \exp(t^i C_i), \nonumber
\end{equation}
where 
\begin{align}
C_2 = & \ -\frac{1}{2} [X,Y]; \nonumber \\
C_3 = & \ \frac{1}{3} [Y,[X,Y]] + \frac{1}{6} [X,[X,Y]]; \nonumber \\
C_4 = & \ -\frac{1}{8} \left ( [Y,[Y,[X,Y]]] + [Y,[X,[X,Y]]] \right ) -\frac{1}{24} [X,[X,[X,Y]]], \nonumber
\end{align}
and higher order terms have increasingly intricate structure. 
While a computer algebra system can evaluate $\exp \left ( -t_{(j,k)} e_{(j,k)}^{(p)} - t_{(\ell,m)} e_{(\ell,m)}^{(p)} \right )$ in closed form, the results involve many pages of arithmetic for the case corresponding to Lemma \ref{lem:pos}, and the other possibilities all yield some manifestly negative entries.


\subsection{\label{sec:BarkerMetropolis}Barker and Metropolis samplers}

Although Lemma \ref{lem:pos} offers only a weak foothold for explicit analytical constructions, we can still use \eqref{eq:explicitsemi1} to produce a MCMC algorithm that is parametrized by $t$. 

\begin{svgraybox}
Here and throughout our discussion of MCMC algorithms, we use a simple trick of relabeling the current state as $n$ and then reversing the relabeling, so that the transition $n \rightarrow j$ becomes generic. 
\end{svgraybox}

For $P = \exp \left (-t e_{(j,j)}^{(p)} \right )$, we have $P_{jj} = 1+f_{(j,j)}^{(p)}(t)$, $P_{jn} = -f_{(j,j)}^{(p)}(t)$, $P_{nj} = -f_{(j,j)}^{(p)}(t)r_j$, and $P_{nn} = 1+f_{(j,j)}^{(p)}(t)r_j$. In particular, 
\begin{equation}
\frac{P_{jn}}{P_{nj}} = \frac{1}{r_j} = \frac{p_n}{p_j}. \nonumber
\end{equation}
That is, detailed balance automatically holds.

The value of $t$ that is optimal for convergence is $t = \infty$, since this maximizes the off-diagonal terms. With this parameter choice, we obtain $P_{jj} = \frac{r_j}{1+r_j}$, $P_{jn} = \frac{1}{1+r_j}$, $P_{nj} = \frac{r_j}{1+r_j}$, and $P_{nn} = \frac{1}{1+r_j}$. The corresponding MCMC algorithm is the so-called Barker sampler. 

However, in light of \eqref{eq:explicitsemi1}, we can almost trivially improve on the Barker sampler. We have that $I - \tau e_{(j,j)}^{(p)} \in \langle p \rangle^+$ iff $0 \le \tau \le \min(1,r_j^{-1})$. But 
\begin{equation}
\label{eq:metropolis}
\left ( I - \min(1,r_j^{-1}) \cdot e_{(j,j)}^{(p)} \right )_{nj} = \min(1,r_j)
\end{equation}
is precisely the Metropolis acceptance ratio. In other words:

\begin{svgraybox}
We have derived the Barker and Metropolis samplers from basic considerations of symmetry and (in the latter case) optimality. 
\end{svgraybox}

\begin{algorithm}[tb]
   \caption{Metropolis}
   \label{alg:metropolis}
\begin{algorithmic}
   \STATE {\bfseries Input:} Runtime $T$ and oracle for $r$
   \STATE Initialize $t=0$ and $X_0$
   \REPEAT
   \STATE Relabel states so that $X_t = n$
   \STATE Propose $j \in [n-1]$
   \STATE Accept $X_{t+1} = j$ with probability \eqref{eq:metropolis}
   \STATE Undo relabeling; set $t = t+1$
   \UNTIL{$t = T$}
   \STATE {\bfseries Output:} $\{X_t\}_{t=0}^T \sim p^{\times (T+1)}$ (approximately)
   \end{algorithmic}
\end{algorithm}

Note that the mechanism for proposing the state $j$ is neither specified nor constrained by our construction. That is, our approach separates concerns between proposal and acceptance mechanisms, and focuses only on the latter. However, a good proposal mechanism is of paramount importance for MCMC algorithms. These observations will continue to apply throughout our later discussion, though in \S \ref{sec:Behavior} we select the elements of proposal sets uniformly at random without replacement for illustrative purposes.

\subsection{\label{sec:Algebra}Some algebra}

The Barker and Metropolis samplers are among the very ``simplest'' MCMC algorithms in that \eqref{eq:explicitsemi1} is among the very sparsest nontrivial matrices in $\langle p \rangle^+$. But if we consider possible transitions to more than one state, we can trade off sparsity for both faster convergence and increased algorithm complexity. The degenerate limiting case is the matrix $1p$, and the practical starting case is the Barker and Metropolis samplers. A central question for interpolating between these cases is how (or if) we can systematically construct denser elements of $\langle p \rangle^+$ than \eqref{eq:explicitsemi1}.

To answer this question, we first generalize Lemma \ref{lem:power}. For $\mathcal{J} := \{j_1,\dots,j_d\} \subseteq [n-1]$ and a matrix $\alpha \in M_{n-1,n-1}$, define $\alpha_{(\mathcal{J})} \in M_{d,d}$ by $(\alpha_{(\mathcal{J})})_{uv} := \alpha_{j_u j_v}$, $\alpha_{(\mathcal{J})}^{(p)} := \sum_{u,v = 1}^d \alpha_{j_u j_v} e_{(j_u,j_v)}^{(p)} \in \mathfrak{lie}(\langle p \rangle)$, and $r_{(\mathcal{J})} := (r_{j_1}, \dots, r_{j_d})$.
\begin{lemma}
\label{lem:abcX} Let $\mathcal{J} := \{j_1,\dots,j_d\} \subseteq [n-1]$. If $\gamma_{(\mathcal{J})}^{(p)} = \alpha_{(\mathcal{J})}^{(p)} \beta_{(\mathcal{J})}^{(p)}$, 
then
\begin{equation}
\label{eq:alphabeta}
\gamma_{(\mathcal{J})} = \alpha_{(\mathcal{J})} (I+1r_{(\mathcal{J})}) \beta_{(\mathcal{J})}.
\end{equation}
\end{lemma}

\begin{proof} 
\begin{align}
\alpha_{(\mathcal{J})}^{(p)} \beta_{(\mathcal{J})}^{(p)} & = \sum_{u,v,w,x} \alpha_{j_u j_v} \beta_{j_w j_x} e_{(j_u,j_v)}^{(p)} e_{(j_w,j_x)}^{(p)}  \nonumber \\
& = \sum_{u,v,w,x} \alpha_{j_u j_v} \left ( \delta_{j_v j_w} + r_{j_w} \right ) \beta_{j_w j_x} e_{(j_u,j_x)}^{(p)} \nonumber \\
& = \sum_{u,x} \left ( \alpha_{(\mathcal{J})} (I+1r_{(\mathcal{J})}) \beta_{(\mathcal{J})} \right )_{ux} e_{(j_u,j_x)}^{(p)}. \nonumber 
\end{align}
where the second equality follows from \eqref{eq:GENproduct} and the third from bookkeeping.
\end{proof}

The heavy notation introduced for Lemma \ref{lem:abcX} is genuinely worthwhile: the case $d = 2$ takes a page to write out by hand without it. More importantly, we can readily construct an analytically convenient matrix in $\mathfrak{lie} ( \langle p \rangle )$ using Lemma \ref{lem:abcX}. 

\begin{theorem}
\label{thm:dxd}
Let $\mathcal{J} := \{j_1,\dots,j_d\} \subseteq [n-1]$, $\omega \in \mathbb{R}$ and
\begin{equation}
\label{eq:dxdA}
A_{(\mathcal{J})}^{(p;\omega)} := \omega \sum_{u,v} \left ( \delta_{j_u j_v} - \frac{1}{1+r_{(\mathcal{J})}1} r_{j_v} \right ) e_{(j_u,j_v)}^{(p)}.
\end{equation}
Then
\begin{equation}
\label{eq:dxd}
\exp tA_{(\mathcal{J})}^{(p;\omega)} = I + \frac{e^{\omega t}-1}{\omega}A_{(\mathcal{J})}^{(p;\omega)}.
\end{equation}
Moreover, $\exp \left ( -tA_{(\mathcal{J})}^{(p;\omega)} \right ) \in \langle p \rangle^+ \cap GL(n, \mathbb{R})$ if $t \ge 0$. In particular, the \emph{Barker matrix}
\begin{equation}
\label{eq:BpJ}
\mathcal{B}_{(\mathcal{J})}^{(p)} := I - A_{(\mathcal{J})}^{(p;1)}
\end{equation}
is in $\langle p \rangle^+$.
\end{theorem}

\begin{proof}
The Sherman-Morrison formula \cite{horn2012matrix} gives that
\begin{equation}
\omega (I+1r_{(\mathcal{J})})^{-1} = \omega \left ( I - \frac{1}{1+r_{(\mathcal{J})}1}1r_{(\mathcal{J})} \right ); \nonumber
\end{equation}
the elements of this matrix are exactly the coefficients in \eqref{eq:dxdA}. Using the notation introduced for the statement of Lemma \ref{lem:abcX}, we can rewrite \eqref{eq:dxdA} as $A_{(\mathcal{J})}^{(p;\omega)} = \left ( \omega (I+1r_{(\mathcal{J})})^{-1} \right )_{(\mathcal{J})}^{(p)}$,
and invoking Lemma \ref{lem:abcX} itself yields $\left ( A_{(\mathcal{J})}^{(p;\omega)} \right )^{i+1} = \omega^i A_{(\mathcal{J})}^{(p;\omega)}$ for $i \in \mathbb{N}$. The result now follows along lines similar to the proof of Theorem \ref{thm:GENbasis}.
\end{proof}

Let $\Delta$ denote the map that sends a matrix to the (column) vector of its diagonal entries, and indicate the boundary of a subset of a topological space using $\partial$.
\begin{lemma}
\label{lem:metropolis}
The \emph{Metropolis matrix}
\begin{equation}
\label{eq:MpJ}
\mathcal{M}_{(\mathcal{J})}^{(p)} := I - \frac{1}{\max \Delta \left ( A_{(\mathcal{J})}^{(p;1)} \right )}A_{(\mathcal{J})}^{(p;1)}
\end{equation} 
is in $\partial \langle p \rangle^+$.
\end{lemma}

\begin{proof}
Writing $A \equiv A_{(\mathcal{J})}^{(p;1)}$ for the moment, the result follows from three basic observations: $\Delta ( A ) \ge 0$, $\max \Delta \left ( A \right ) > 0$, and $A - \Delta \left ( \Delta \left ( A \right ) \right ) \le 0$.
\end{proof}

\subsubsection{\label{sec:Example}Example}

    %
    %
    %
    %
    %

To illustrate the Barker and Metropolis matrix constructions, consider $\mathcal{J} = \{1,2,3\}$ and $p = (1,2,3,4,10)/20$. Now \eqref{eq:dxdA} is
\begin{equation}
A_{(\mathcal{J})}^{(p;\omega)} = 
\frac{\omega}{16}\begin{pmatrix} 15 & -2 & -3 & 0 & -10 \\ -1 & 14 & -3 & 0 & -10 \\ -1 & -2 & 13 & 0 & -10 \\ 0 & 0 & 0 & 0 & 0 \\ -1 & -2 & -3 & 0 & 6 \end{pmatrix}. \nonumber
\end{equation}
For $\omega = 1$ and $t = -\log 2$, \eqref{eq:dxd} is  
\begin{equation}
\exp \left ( \log 2 \cdot A_{(\mathcal{J})}^{(p;1)} \right ) = 
\frac{1}{32}\begin{pmatrix} 17 & 2 & 3 & 0 & 10 \\ 1 & 18 & 3 & 0 & 10 \\ 1 & 2 & 19 & 0 & 10 \\ 0 & 0 & 0 & 32 & 0 \\ 1 & 2 & 3 & 0 & 26 \end{pmatrix}. \nonumber
\end{equation}
whereas for $\omega = 2$ and $t = -\log 2$, \eqref{eq:dxd} is
\begin{equation}
\exp \left ( \log 2 \cdot A_{(\mathcal{J})}^{(p;2)} \right ) = 
\frac{1}{64}\begin{pmatrix} 19 & 6 & 9 & 0 & 30 \\ 3 & 22 & 9 & 0 & 30 \\ 3 & 6 & 25 & 0 & 30 \\ 0 & 0 & 0 & 64 & 0 \\ 3 & 6 & 9 & 0 & 46 \end{pmatrix}. \nonumber
\end{equation}
Finally, \eqref{eq:BpJ} and \eqref{eq:MpJ} are respectively
\begin{equation}
\mathcal{B}_{(\mathcal{J})}^{(p)} =
\frac{1}{16}\begin{pmatrix} 1 & 2 & 3 & 0 & 10 \\ 1 & 2 & 3 & 0 & 10 \\ 1 & 2 & 3 & 0 & 10 \\ 0 & 0 & 0 & 16 & 0 \\ 1 & 2 & 3 & 0 & 10 \end{pmatrix}; \quad \mathcal{M}_{(\mathcal{J})}^{(p)} =
\frac{1}{15}\begin{pmatrix} 0 & 2 & 3 & 0 & 10 \\ 1 & 1 & 3 & 0 & 10 \\ 1 & 2 & 2 & 0 & 10 \\ 0 & 0 & 0 & 15 & 0 \\ 1 & 2 & 3 & 0 & 9 \end{pmatrix}. \nonumber
\end{equation}

\subsection{\label{sec:HigherOrder}Higher-order samplers}

In order to obtain higher-order samplers from the algebra of \S \ref{sec:Algebra}, we use a familiar trick, letting $n \rightarrow j \in \mathcal{J}$ correspond to a generic transition as in \S \ref{sec:BarkerMetropolis}. (Again, we do not specify or constrain a mechanism for proposing a set $\mathcal{J}$ of candidate states to transition into.) This immediately yields more sophisticated MCMC algorithms using \eqref{eq:BpJ} and \eqref{eq:MpJ} which we call \emph{higher-order Barker and Metropolis samplers}, respectively abbreviated as HOBS and HOMS. 

The corresponding matrix elements are straightforwardly obtained with a bit of arithmetic:
\begin{align}
\label{eq:entriesA}
\frac{1}{\omega} \left ( A_{(\mathcal{J})}^{(p;\omega)} \right )_{j_u j_u} & = 1 - \frac{r_{j_u}}{1+r_{(\mathcal{J})}1}; \nonumber \\
\frac{1}{\omega} \left ( A_{(\mathcal{J})}^{(p;\omega)} \right )_{n j_u} & = - \frac{r_{j_u}}{1+r_{(\mathcal{J})}1}; \nonumber \\
\frac{1}{\omega} \left ( A_{(\mathcal{J})}^{(p;\omega)} \right )_{n n} & = \frac{r_{(\mathcal{J})}1}{1+r_{(\mathcal{J})}1},
\end{align}
which yields the HOBS:
\begin{align}
\label{eq:HOBS}
\left ( \mathcal{B}_{(\mathcal{J})}^{(p)} \right )_{n j_u} & = \frac{r_{j_u}}{1+r_{(\mathcal{J})}1}; \nonumber \\
\left ( \mathcal{B}_{(\mathcal{J})}^{(p)} \right )_{n n} & = \frac{1}{1+r_{(\mathcal{J})}1}.
\end{align}

Meanwhile, 
\begin{equation}
\frac{1}{\omega} \max \Delta \left ( A_{(\mathcal{J})}^{(p;\omega)} \right ) = \frac{1+r_{(\mathcal{J})}1-\min \{ 1, \min r_{(\mathcal{J})} \}}{1+r_{(\mathcal{J})}1} \nonumber
\end{equation}
yielding the HOMS:
\begin{align}
\label{eq:HOMS}
\left ( \mathcal{M}_{(\mathcal{J})}^{(p)} \right )_{n j_u} & = \frac{r_{j_u}}{1+r_{(\mathcal{J})}1-\min \{ 1, \min r_{(\mathcal{J})} \}}; \nonumber \\
\left ( \mathcal{M}_{(\mathcal{J})}^{(p)} \right )_{n n} & = 1 - \frac{r_{(\mathcal{J})}1}{1+r_{(\mathcal{J})}1-\min \{ 1, \min r_{(\mathcal{J})} \}}.
\end{align}

\begin{algorithm}[tb]
   \caption{HOMS}
   \label{alg:homs}
\begin{algorithmic}
   \STATE {\bfseries Input:} Runtime $T$ and oracle for $r$
   \STATE Initialize $t=0$ and $X_0$
   \REPEAT
   \STATE Relabel states so that $X_t = n$
   \STATE Propose $\mathcal{J} = \{j_1,\dots,j_d\} \subseteq [n-1]$
   \STATE Accept $X_{t+1} = j_u$ with probability \eqref{eq:HOMS}
   \STATE Undo relabeling; set $t = t+1$
   \UNTIL{$t = T$}
   \STATE {\bfseries Output:} $\{X_t\}_{t=0}^T \sim p^{\times (T+1)}$ (approximately)
   \end{algorithmic}
\end{algorithm}

It turns out that the HOBS is equivalent to the ensemble MCMC algorithm of \cite{neal2011ensemble} as described in \cite{martino2018MT,martino2018IRSM}. The proposal mechanism we use for the HOBS in \S \ref{sec:Behavior} essentially amounts to the \emph{independent} ensemble MCMC sampler (apart from non-replacement, which technically induces jointness), but in general this is not the case. A more sophisticated proposal mechanism that can exploit any joint structure in the target distribution would be more powerful, but we reiterate that our approach is completely agnostic to proposal mechanism details.

In contrast, the HOMS is different than a \emph{multiple-try Metropolis sampler} (MTMS), including the independent MTMS described in \cite{martino2018MT}. The HOMS uses a sample from $\mathcal{J} \cup \{n\}$ to perform a state transition in a single step according to \eqref{eq:HOMS}, whereas a MTMS first samples from $\mathcal{J}$ before accepting or rejecting the result. The HOMS (and for that matter, also the HOBS) actually turns out to be a slightly special case of a construction in \S 2.3 of \cite{delmas2009does}. This work uses a ``proposition kernel'' defined by assigning a probability distribution on the power set $2^{[n]}$ of the state space $[n]$ to each element of the state space. Essentially, the HOMS and HOBS result if this distribution on $2^{[n]}$ is independent of the individual element (i.e., it varies only with the subset).

\subsection{\label{sec:Behavior}Behavior of higher-order samplers}

The difference between the HOBS and HOMS decreases as $d = |\mathcal{J}|$ increases and/or $p$ becomes less uniform (e.g., in a low-temperature limit), since in either limit we have $\min \{ 1, \min r_{(\mathcal{J})} \} \ll 1+r_{(\mathcal{J})}1$. Although one might hope to gain the most benefit from improved MCMC algorithms in such situations, the HOMS can still provide a comparative advantage for $d > 1$ but small, with elements chosen in complementary ways (uniformly at random, near current/previous states, etc.), or in e.g. the high-temperature part of a parallel tempering scheme \cite{earl2005tempering}.

We use the example of a small Sherrington-Kirkpatrick (SK) spin glass \cite{bolthausen2007spin,panchenko2012sk} to exhibit the behavior of the HOBS and HOMS in 
Figures \ref{fig:SK_9spins_b025} and \ref{fig:SK_9spins_b100}. The SK spin glass is the distribution \begin{equation}
\label{eq:SK}
p(s) := \textstyle{Z^{-1} \exp \left ( -\frac{\beta}{\sqrt N} \sum_{jk} J_{jk} s_j s_k \right )}
\end{equation}
over spins $s \in \{\pm 1\}^N$, where $J$ is a symmetric $N \times N$ matrix with independent identically distributed standard Gaussian entries and $\beta$ is the inverse temperature. 

\begin{figure}[htbp]
\includegraphics[trim = 10mm 0mm 10mm 0mm, clip, width=\textwidth,keepaspectratio]{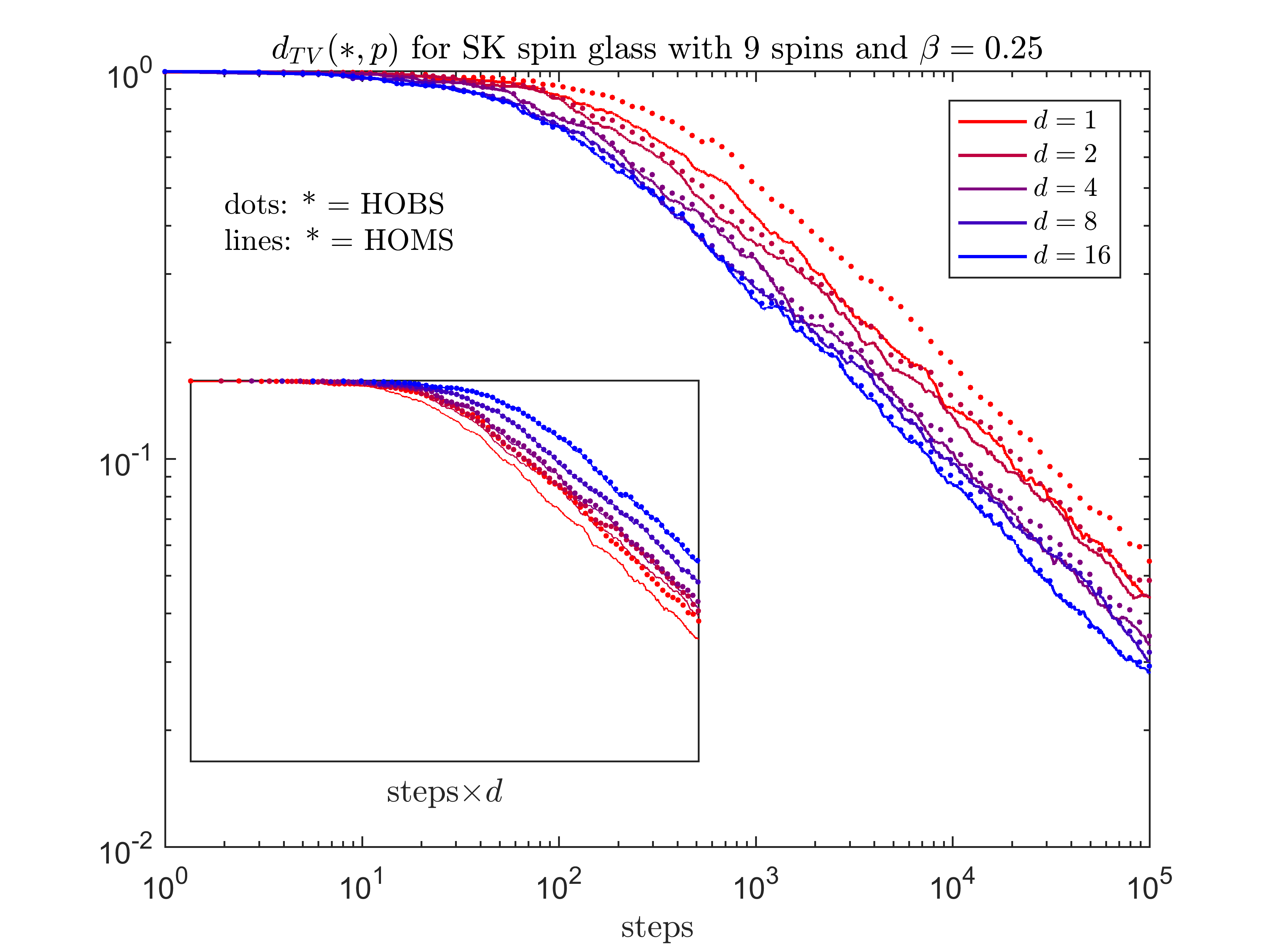}
\caption{ \label{fig:SK_9spins_b025} Total variation distance between the HOBS/HOMS with proposal sets $\mathcal{J}$ (elements distributed uniformly without replacement) of varying sizes $d$ and \eqref{eq:SK} with 9 spins and $\beta = 1/4$. Inset: the same data and window, with horizontal axis normalized by $d$.
} 
\end{figure} %


\begin{figure}[htbp]
\includegraphics[trim = 10mm 0mm 10mm 0mm, clip, width=\textwidth,keepaspectratio]{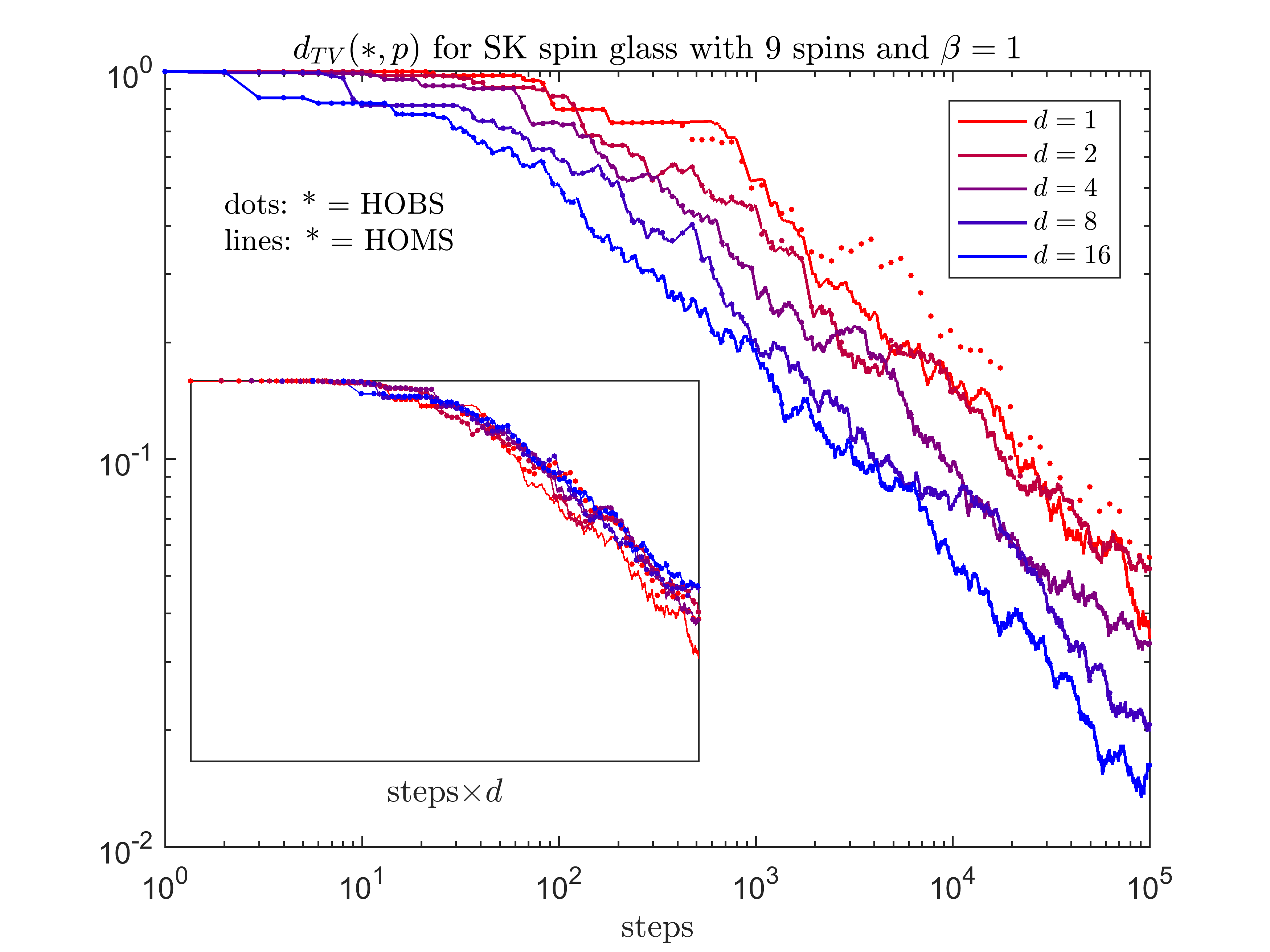}
\caption{ \label{fig:SK_9spins_b100} As in Figure \ref{fig:SK_9spins_b025} with $\beta = 1$.
} 
\end{figure} %

The SK model is well-suited for a straightforward evaluation of higher-order samplers because of its disordered energy landscape. More detailed models or benchmarks seem to require specific assumptions (e.g., the particular form of a spin Hamiltonian for Swendsen-Wang updates) and/or parameters (e.g., additional temperatures for parallel tempering, or of a vorticity matrix for non-reversible Metropolis-Hastings). In keeping with a straightforward evaluation, we do not consider sophisticated or diverse ways to generate elements of proposal sets $\mathcal{J}$. Instead, we simply select elements of $\mathcal{J}$ uniformly at random without replacement. We use the same pseudorandom number generator initial state for all simulations in order to highlight relative behavior. Finally, we choose $\beta$ low enough ($1/4$ and $1$) so that the behavior of a single run is sufficiently representative to make simple qualitative judgments.

The figure insets show that although higher-order samplers indeed converge more quickly, this comes at the cost of more overall evaluations of probability ratios. Parallelism is therefore necessary for higher-order samplers to be a wise algorithmic choice.

We reiterate in closing this section that the HOMS gives results very close to the HOBS, except for small values of $d$ or a more uniform target distribution $p$. Increasing the number of spins in the SK model and/or considering an Edwards-Anderson spin glass also yields qualitatively similar results (not shown here).

\subsection{\label{sec:LP}Linear objectives for transition matrices}

\begin{svgraybox}
We can push the preceding ideas further by using an optimization scheme to construct transition matrices with the desired invariant measures and that saturate a suitable objective function. 
\end{svgraybox}

For example, the linear objective $-1_\mathcal{J}^T \tau_{(\mathcal{J})}^{(p)} r_\mathcal{J}^T$ considered immediately after \eqref{eq:genericObjective} yields an optimal sparse approximation of the ``ultimate'' transition matrix $1p$. (To the best of our knowledge, this construction has not been considered elsewhere.) However, bringing an optimization scheme to bear narrows the regime of applicability to cases where computing likelihoods is hard enough and sufficient parallel resources are available to justify the added computational costs. 

To make this concrete, first define $1_\mathcal{J} \in \mathbb{R}^n$ by 
\begin{equation}
(1_\mathcal{J})_j := \begin{cases} 1 & \text{if } j \in \mathcal{J} \cup \{n\} \\ 0 & \text{otherwise}, \end{cases} \nonumber
\end{equation}
$1_\mathcal{J}^- := ((1_\mathcal{J})_1,\dots,(1_\mathcal{J})_{n-1})^T$, $r_\mathcal{J} := r \odot 1_\mathcal{J}^T$, and $r_\mathcal{J}^- := r^- \odot (1_\mathcal{J}^-)^T$, where $\odot$ is the entrywise or Hadamard product (note that $r_\mathcal{J} \in \mathbb{R}^n$, while $r_{(\mathcal{J})} \in \mathbb{R}^{|\mathcal{J}|}$ has been defined previously). 

Write $\Delta$ for the matrix diagonal map and recall the notation of Lemma \ref{lem:abcX}: since
\begin{equation}
\label{eq:tau}
\tau_{(\mathcal{J})}^{(p)} = \begin{pmatrix} I_{n-1} \\ -r_\mathcal{J}^- \end{pmatrix} \tau \begin{pmatrix} I_{n-1} & -1_\mathcal{J}^- \end{pmatrix},
\end{equation}
we have that $I - \tau_{(\mathcal{J})}^{(p)} \in \langle p \rangle^+$ iff
\begin{subequations}
\label{eq:constraints}
\begin{align}
0 & \le I_{n-1} - \Delta(1_\mathcal{J}^-) \tau \Delta(1_\mathcal{J}^-) \le 1; \label{eq:bodyConstraint} \\
0 & \le \tau 1_\mathcal{J}^- \le 1; \label{eq:lastColConstraint} \\
0 & \le r_\mathcal{J}^- \tau \le 1; \label{eq:lastRowConstraint} \\
0 & \le r_\mathcal{J}^- \tau 1_\mathcal{J}^- \le 1. \label{eq:lastElementConstraint}
\end{align}
\end{subequations}
The constraints \eqref{eq:lastColConstraint}-\eqref{eq:lastElementConstraint} respectively force the first $n-1$ entries of the last column, the first $n-1$ entries of the last row, and the bottom right matrix entry of $\tau_{(\mathcal{J})}^{(p)}$ to be in the unit interval. \eqref{eq:bodyConstraint} forces the relevant entries of the ``coefficient matrix'' $\tau$ (as an upper left submatrix of $\tau_{(\mathcal{J})}^{(p)}$) to be in the unit interval.

We can conveniently set to zero the irrelevant/unspecified rows and columns of $\tau$ that do not contribute to $\tau_{(\mathcal{J})}^{(p)}$ via the constraints
\begin{equation}
\label{eq:sparsityConstraint}
\Delta(1-1_\mathcal{J}^-) \tau = \tau \Delta(1-1_\mathcal{J}^-) = 0.
\end{equation}
Provided that we impose \eqref{eq:sparsityConstraint}, \eqref{eq:bodyConstraint} can be replaced with
\begin{equation}
\label{eq:bodyConstraint2}
0 \le I_{n-1} - \tau \le 1.
\end{equation}

The ``diagonal'' case corresponding to Lemma \ref{lem:pos} shows that \eqref{eq:constraints} and \eqref{eq:sparsityConstraint} jointly have nontrivial solutions. This suggests that we consider suitable objectives and corresponding linear programs for optimizing the MCMC transition matrix $I - \tau_{(\mathcal{J})}^{(p)}$. We therefore introduce the \emph{vectorization} map $\text{vec}$ that sends a matrix to a vector by stacking matrix columns in order. This map obeys the useful identity $\text{vec}(XYZ^T) = (Z \otimes X) \text{vec}(Y)$, where $\otimes$ denotes the Kronecker product.

Now a reasonably generic linear objective function is
\begin{equation}
\label{eq:genericObjective}
 x^T \tau_{(\mathcal{J})}^{(p)} y = (y^T \otimes x^T) \text{vec} \left ( \tau_{(\mathcal{J})}^{(p)} \right )
\end{equation}
for suitable fixed $x$ and $y$. In practice, we consider $x = 1_\mathcal{J}$ and $y = -r_\mathcal{J}^T$. This maximizes the Frobenius inner product of $I - \tau_{(\mathcal{J})}^{(p)}$ and $1_\mathcal{J} r_\mathcal{J}$ because
\begin{equation}
\text{Tr} \left ( \left ( I - \tau_{(\mathcal{J})}^{(p)} \right )^T 1_\mathcal{J} r_\mathcal{J} \right ) = r_\mathcal{J} 1_\mathcal{J} - 1_\mathcal{J}^T \tau_{(\mathcal{J})}^{(p)} r_\mathcal{J}^T. \nonumber
\end{equation}
Alternatives like $x = e_n, y = e_n$ (to discourage self-transitions) can lead to convergence that slows catastrophically as $d = |\mathcal{J}|$ increases, because high-probability states are less likely to remain occupied. More surprisingly, the same sort of slowing down happens for $x = e_n, y = -r_\mathcal{J}^T$, as well as for variations involving the $n$th component of $y$. We suspect that the cause is the same, albeit mediated indirectly through an objective that ``overfits'' the proposed transition probabilities to the detriment of remaining in place (or in some cases ``underfits'' by producing the identity matrix). Overall, it appears nontrivial to select better choices for $x$ and $y$ than our defaults above.

By \eqref{eq:tau} we get
\begin{equation}
\text{vec} \left ( \tau_{(\mathcal{J})}^{(p)} \right ) = \left [ \begin{pmatrix} I_{n-1} \\ -(1_\mathcal{J}^-)^T \end{pmatrix} \otimes \begin{pmatrix} I_{n-1} \\ -r_\mathcal{J}^- \end{pmatrix} \right ] \text{vec}(\tau),
\end{equation}
and in turn $(y^T \otimes x^T) \text{vec} \left ( \tau_{(\mathcal{J})}^{(p)} \right )$ equals
\begin{equation}
\label{eq:tensorObjective}
\left [ y^T \begin{pmatrix} I_{n-1} \\ -(1_\mathcal{J}^-)^T \end{pmatrix} \otimes x^T \begin{pmatrix} I_{n-1} \\ -r_\mathcal{J}^- \end{pmatrix} \right ] \text{vec}(\tau).
\end{equation}

Now the constraints and the objective of the linear program are both explicitly specified in terms of the ``coefficient'' matrix $\tau$, so in principle we have a working algorithm already. However, it is convenient to respectively rephrase the constraints \eqref{eq:lastColConstraint}-\eqref{eq:lastElementConstraint}, \eqref{eq:sparsityConstraint}, and \eqref{eq:bodyConstraint2} into different forms as
\begin{equation}
0 \le \begin{pmatrix} \left ( 1_\mathcal{J}^- \right )^T \otimes I_{n-1} \\ I_{n-1} \otimes r_\mathcal{J}^- \\ \left ( 1_\mathcal{J}^- \right )^T \otimes r_\mathcal{J}^- \end{pmatrix} \text{vec}(\tau) \le 1,
\end{equation}
\begin{equation}
\begin{pmatrix} I_{n-1} \otimes \Delta(1-1_\mathcal{J}^-) \\ \Delta(1-1_\mathcal{J}^-) \otimes I_{n-1} \end{pmatrix} \text{vec}(\tau) = 0,
\end{equation}
\begin{equation}
\text{vec}(I_{n-1})-1 \le \text{vec}(\tau) \le \text{vec}(I_{n-1}).
\end{equation}

Therefore, writing
\begin{align}
U_{(\mathcal{J})}^{(p)} := & \ \begin{pmatrix} I_{2n-1} \\ -I_{2n-1} \end{pmatrix} \begin{pmatrix} \left ( 1_\mathcal{J}^- \right )^T \otimes I_{n-1} \\ I_{n-1} \otimes r_\mathcal{J}^- \\ \left ( 1_\mathcal{J}^- \right )^T \otimes r_\mathcal{J}^- \end{pmatrix}; \nonumber \\
v := & \ \begin{pmatrix} 1_{2n-1} \\ 0_{2n-1} \end{pmatrix}; \nonumber \\
w_{(\mathcal{J})}^{(p)} := & \ -y^T \begin{pmatrix} I_{n-1} \\ -(1_\mathcal{J}^-)^T \end{pmatrix} \otimes x^T \begin{pmatrix} I_{n-1} \\ -r_\mathcal{J}^- \end{pmatrix}, \nonumber
\end{align}
and 
\begin{equation}
U_{(\mathcal{J})}^{(0)} := \begin{pmatrix} I_{n-1} \otimes \Delta(1-1_\mathcal{J}^-) \\ \Delta(1-1_\mathcal{J}^-) \otimes I_{n-1} \end{pmatrix},
\end{equation}
we can at last write the sought-after linear program (noting a minus sign included in $w_{(\mathcal{J})}^{(p)}$) in a form suitable for (e.g.) MATLAB's \texttt{linprog} solver:
\begin{subequations}
\label{eq:LP}
\begin{align}
\min_{\tau} w_{(\mathcal{J})}^{(p)} \text{vec}(\tau) \quad \text{s.t.}
	& \nonumber \\
U_{(\mathcal{J})}^{(p)} \text{vec}(\tau) \quad \le 
	& \quad v; \label{eq:LPineq} \\
U_{(\mathcal{J})}^{(0)} \text{vec}(\tau) \quad = 
	& \quad 0; \label{eq:LPeq} \\
\text{vec}(\tau) \quad \ge 
	& \quad \text{vec}(I_{n-1})-1; \label{eq:LPlb} \\
\text{vec}(\tau) \quad \le 
	& \quad \text{vec}(I_{n-1}). \label{eq:LPub}
\end{align}
\end{subequations}

The preceding discussion therefore culminates in the following
\begin{theorem}
\label{thm:LP}
The linear program \eqref{eq:LP} has a solution in $\langle p \rangle^+$. \qed
\end{theorem}

\subsubsection{\label{sec:LPexample}Example}


As in \S \ref{sec:Example}, consider $p = (1,2,3,4,10)/20$ and $\mathcal{J} = \{1,2,3\}$. Solving the linear program with $x = 1_\mathcal{J}$ and $y = -r_\mathcal{J}^T$ produces the following element of $\langle p \rangle^+$:
\begin{equation}
\begin{pmatrix} 0 & 0 & 0 & 0 & 1 \\ 0 & 0 & 0 & 0 & 1 \\ 0 & 0 & 0 & 0 & 1 \\ 0 & 0 & 0 & 1 & 0 \\ 0.1 & 0.2 & 0.3 & 0 & 0.4 \end{pmatrix}. \nonumber
\end{equation}
For comparison, recall that the last row of $\mathcal{M}_{(\mathcal{J})}^{(p)}$ equals 
$(0.0 \bar 6, 0.1 \bar 3, 0.2, 0, 0.6)$.

\subsubsection{\label{sec:HOPS}The higher-order programming sampler}

We call the sampler obtained from \eqref{eq:genericObjective} and \eqref{eq:LP} with $x = -1_\mathcal{J}$ and $y = r_\mathcal{J}^T$ the \emph{higher-order programming sampler} (HOPS). We compare the HOMS and HOPS in Figures \ref{fig:SK_9spins_b025_LPcomparison} and \ref{fig:SK_9spins_b100_LPcomparison} (cf. Figures \ref{fig:SK_9spins_b025} and \ref{fig:SK_9spins_b100}). The figures show that the HOPS improves upon the HOMS, which in turn improves upon the HOBS.

\begin{figure}[htbp]
\includegraphics[trim = 10mm 0mm 10mm 0mm, clip, width=\textwidth,keepaspectratio]{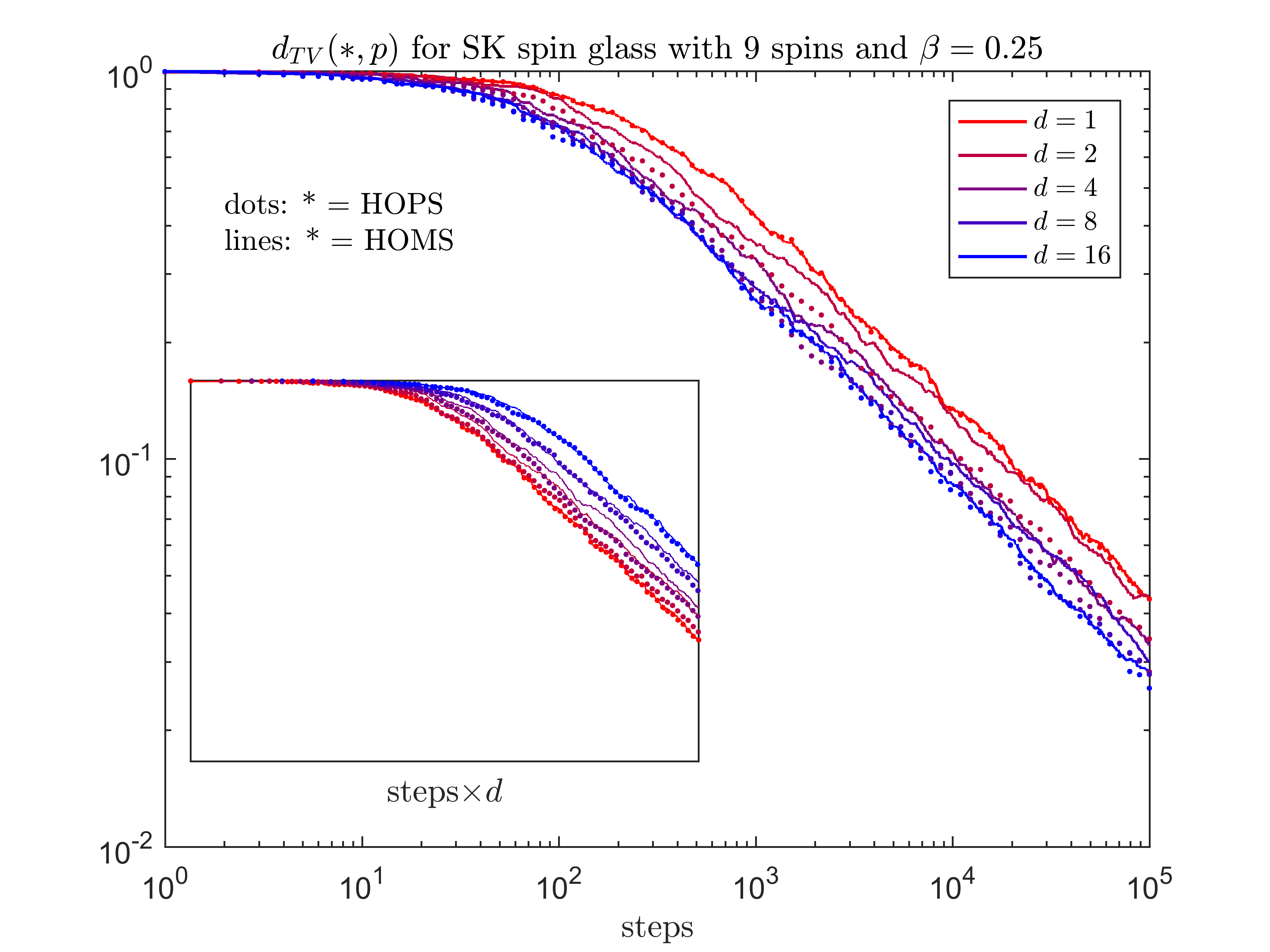}
\caption{ \label{fig:SK_9spins_b025_LPcomparison} Total variation distance between the HOPS/HOMS with proposal sets $\mathcal{J}$ (elements sampled uniformly without replacement) of varying sizes $d$ and \eqref{eq:SK} with 9 spins and $\beta = 1/4$. Inset: same data and window, with horizontal axis normalized by $d$.
} 
\end{figure} %

\begin{figure}[htbp]
\includegraphics[trim = 10mm 0mm 10mm 0mm, clip, width=\textwidth,keepaspectratio]{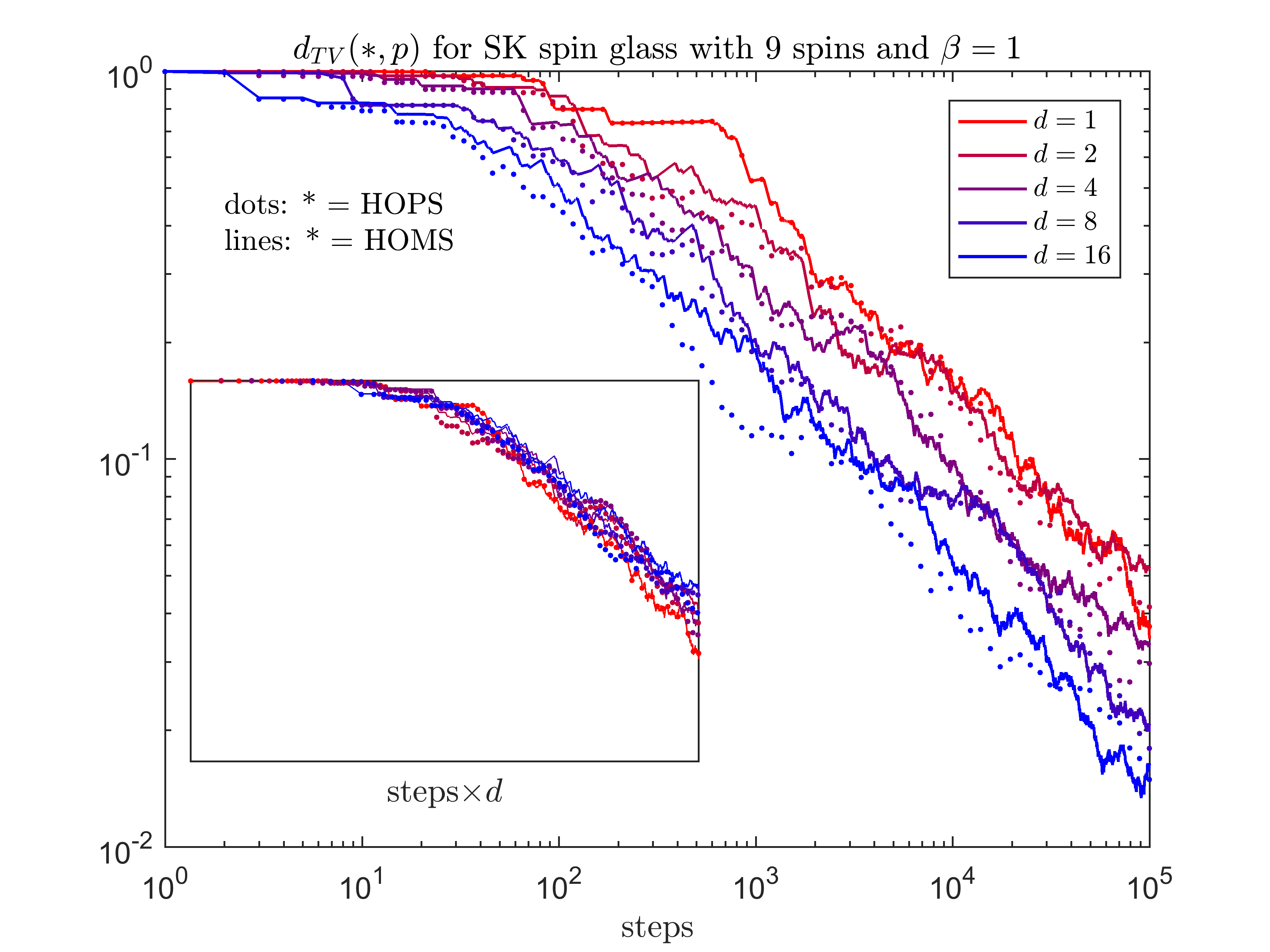}
\caption{ \label{fig:SK_9spins_b100_LPcomparison} As in Figure \ref{fig:SK_9spins_b025_LPcomparison} with $\beta = 1$.
} 
\end{figure} %

\begin{algorithm}[tb]
   \caption{HOPS}
   \label{alg:hops}
\begin{algorithmic}
   \STATE {\bfseries Input:} Runtime $T$ and oracle for $r$
   \STATE Initialize $t=0$ and $X_0$
   \REPEAT
   \STATE Relabel states so that $X_t = n$
   \STATE Propose $\mathcal{J} = \{j_1,\dots,j_d\} \subseteq [n-1]$
   \STATE Compute $\tau$ solving \eqref{eq:LP} with $x = 1_\mathcal{J}$ and $y = -r_\mathcal{J}^T$
   \STATE Set $P = I-\tau_{(\mathcal{J})}^{(p)}$ using \eqref{eq:tau}
   \STATE Accept $X_{t+1} = j_u$ with probability $P_{nj_u}$
   \STATE Undo relabeling; set $t = t+1$
   \UNTIL{$t = T$}
   \STATE {\bfseries Output:} $\{X_t\}_{t=0}^T \sim p^{\times (T+1)}$ (approximately)
   \end{algorithmic}
\end{algorithm}

\subsection{\label{sec:RemarksPhysics}Remarks on sampling}

Besides providing a framework that conceptually unifies various MCMC algorithms, symmetry principles lead to the apparently new HOPS algorithm of \S \ref{sec:LP}. It is possible that the HOPS itself might be further improved upon by developing an objective function suited for, e.g. convex optimization versus a mere linear program. These ideas might also enhance existing MCMC techniques specifically tailored for parallel computation, as in \cite{conrad2018expensive}. In particular, the Bayesian approach to inverse problems \cite{dashti2016inverse} may be fertile ground for applications. 

As we have already mentioned, our approach is agnostic with respect to proposals, focusing purely on acceptance mechanisms. However, the proposal mechanism has less impact than the acceptance mechanism in practice, especially for differentiable distributions. In practice, a stateful and/or problem-specific proposal exploiting joint structure is highly desirable and even necessary for any real utility, but we these avenues unexplored for now (one possibility is suggested by particle MTMS algorithms as in \cite{martino2014} and exploiting tensor product structure in transition matrices and $\langle p \rangle$). It would be of interest to incorporate some aspect of a proposal mechanism into the objective of \eqref{eq:LP}, but it is not clear how to actually do this. In fact, our numerical example featured a SK spin glass to illustrate our ideas precisely because its highly disordered structure (and discrete state space) are suited for separating concerns about proposal and acceptance mechanisms.

It would certainly be interesting to extend the present considerations to continuous variables. However, this would seem to require a more technical treatment, since infinite-dimensional Lie theory, distributions \emph{\`a la} Schwartz, etc. would play a role at least in principle. In a complementary vein, it would be interesting to see if the full construction of \cite{delmas2009does} could be recovered from symmetry arguments alone.

While the Barker and Metropolis samplers are reversible, it is not clear if the HOPS is, though \cite{bierkens2016nonreversible} points out ways to transform reversible kernels into irreversible ones and \emph{vice versa}.

It is possible to produce transition matrices (even in closed form) in which the $n$th row is nonnegative but other rows have negative entries. It is not immediately clear if using such a matrix actually poisons a MCMC algorithm. Though preliminary experiments in this direction were discouraging, we have not found a compelling argument that rules out the use of such matrices. 

Finally, it would be of interest to sample from the vertices of the polytope $\langle p \rangle^+$. However, (even approximately) uniformly sampling vertices of a polytope is $\mathbf{NP}$-hard (and thus presumably intractable) by Theorem 1 of \cite{khachiyan2001transversal}: see also \cite{khachiyan2008vertices}.

\section{\label{sec:2} Statistical physics via symmetry}

We have seen in \S \ref{sec:1} that sampling algorithms can be better understood in principle and also accelerated in practice through elementary considerations of symmetry. In the present section, we show how similarly basic considerations of symmetry can derive the basic structure of statistical physics. While we do not address entropy \emph{per se}, that ground is well-traveled, with the well-known characterization of Faddeev \cite{faddeev1956concept,baez2011characterization} playing an exemplary role. 

We focus instead on the role of temperature (and via closure of the Gibbs relation, energy), which classical information-theoretical considerations have not substantially accounted for. In particular, we sketch how an effective temperature can reproduce the physical temperature for conjecturally generic model systems (see also \S \ref{sec:Anosov}), while also enabling applications to data analytics, characterization of time-inhomogeneous Markov processes, nonequilibrium thermodynamics, etc.

The goal of providing a self-consistent description of stationary systems with finitely many states using the language of equilibrium statistical physics in the canonical ensemble naturally flows from the idea expressed in \cite{gallavotti2008heat} that ``there is no conceptual difference between stationary states in equilibrium and out of equilibrium.'' While the traditional aim of statistical physics is predicting statistical behavior in terms of measured physical properties, the aim here is to go in the other direction: that is, to determine effective physical properties--in and out of equilibrium--in terms of observable statistical behavior. In other words, the goal is to take one step farther the now-classical maximum entropy point of view in which statistical physics is a framework for reasoning about data.

We realize this goal by demonstrating the existence, uniqueness (up to a choice of scale), and relevance of a physically reasonable and \emph{invertible} transformation between simple effective statistical and physical descriptions of a system (see figure \ref{fig:t2H}). The effective statistical description is furnished by a probability distribution along with a characteristic timescale. The effective physical description consists of an effective energy function and an effective temperature. 
\footnote{
The use of an effective temperature in glassy systems has a long history \cite{tool1946relation,nieuwenhuizen1998thermodynamics,leuzzi2007thermodynamics} and has recently gained prominence through the \emph{fluctuation-dissipation (FD) temperature} in mean-field systems \cite{cugliandolo2011effective,puglisi2017temperature}. Discussions of the relationship between our construction and both the FD temperature (frequently called ``the'' effective temperature in the literature) and the dynamical temperature introduced by Rugh \cite{rugh1997dynamical,rugh1998geometric} can be found in \cite{huntsman2010anosov}. 
}
The transformation between these descriptions will be derived from the elementary Gibbs relation and basic symmetry considerations along lines first explored in \cite{ford2005surfaces,ford2006descriptive}. 

The utility and naturalness of the effective physical description that results from performing this transformation on an effective statistical description will depend entirely on the utility and naturalness of the underlying state space and of the characteristic timescale. In the event that the actual state space of a real physical system in thermal equilibrium is finite and an appropriate characteristic timescale can be determined, the corresponding effective physical description will manifestly reproduce the actual physics. Moreover, in near-equilibrium, the effective temperature and energies will remain near the actual values of their equilibrium analogues by a continuity argument. Consequently, the framework discussed here may inform principled characterizations of quasi-equilibria.

However, as the system is driven away from equilibrium, its effective energy levels will shift, while the actual energy levels of a real physical system may be fixed and intrinsic. Nevertheless, such shifts are still of interest for characterizing nonequilibrium situations, even for real physical systems. For example, a system such as a laser undergoing population inversion will exhibit an effective level crossing as the driving parameter varies. In a related vein, a negative absolute temperature \cite{braun2013negative,dunkel2014consistent,frenkel2015gibbs} would correspond in our framework to a negative characteristic timescale, indicating antithermodynamic behavior such as ``antimixing'' or ``antirelaxation.''

Even very limited knowledge about the energy levels and temperature of a system is sufficient to determine the remainder of that information as a trivial exercise in algebra using the Gibbs relation. Nevertheless, the preceding discussion should not distract from the observation that the framework discussed here provides its most substantial advantage in the situation where inverting the Gibbs relation might initially seem like an ill-posed problem. Therefore, the primary goal of the framework discussed below is to give effective physical descriptions of systems that have no \emph{a priori} physical characterization, while maintaining total consistency with equilibrium statistical physics in situations where an \emph{a priori} physical characterization \emph{is} available.

Highlighting this consistency is the example of Anosov systems (see \S \ref{sec:Anosov}), and specifically paradigmatic chaotic model systems such as the cat map and the free particle or ideal gas on a surface of constant negative curvature, where using a careful iterative discretization scheme indicates how the actual energy and temperature may be reproduced by suitable effective analogues, despite the fact that the underlying state spaces are continuous. Thermostatting subsequently indicates how the transformation at the heart of our discussion could be applied in principle to essentially arbitrary physical systems \cite{huntsman2010anosov}. 

While the perspective we shall take below does not confer extensive predictive power in the realm of physics, it does have some (see, e.g. \S \ref{sec:TwoStateSystems}) and its descriptive and explanatory power nevertheless suggests a wide and significant scope for applications, including to nonequilibrium statistical physics, the renormalization group, information theory, and the characterization of both stochastic processes and experimental data. More provocatively, it can be regarded as illuminating the fundamental meaning of both energy and temperature independently of references to work, force, mass, or the underlying spatial context upon which the latter concepts ultimately depend for their definition. 

In this section, we derive the Gibbs relation from symmetry considerations in \S \ref{sec:GibbsDerivation} before introducing the coordinate systems that respectively underlie experimental/probabilistic and theoretical/physical descriptions of systems in \S \ref{sec:Descriptions}. With the stage set, we perform some preliminary algebra in \S \ref{sec:PrelliminaryAlgebra}. After obtaining intermediate results on the scaling behavior of inverse temperature as a function of time in \S \ref{sec:Scaling} and on the geometry of any reasonable transformation between the two descriptions above in \S \ref{sec:Geometry}, we complete the derivation of the effective temperature in \S \ref{sec:EffectiveTemperature}. We then outline constraints on the form of a characteristic timescale imposed by considering product systems in \S \ref{sec:ProductSystems}. Finally, we outline a number of examples and applications in \S \ref{sec:ExamplesAndApplications} before remarks in \S \ref{sec:RemarksStatPhys}. 

Later, \S \ref{sec:Anosov} considers the effective temperature for Anosov systems.

\begin{svgraybox}
At times, we may write $\beta$ to denote the \emph{physical or actual inverse temperature} as well as an effective analogue. Context should serve to eliminate any ambiguity, especially as we make an effort to separate discussion of these two quantities impinging on equations.
\end{svgraybox}

\subsection{\label{sec:GibbsDerivation}The Gibbs distribution}

The first step in deriving the basic structure of statistical physics from symmetry is to derive the Gibbs relation between state probabilities and energies. We do this for a finite system from the basic postulate that the probability of a state depends only on its energy. 
\footnote{
In a similar if slightly less parsimonious vein, Blake Stacey has pointed out that the Gibbs distribution can be derived ``based on the idea that if [two systems] $A$ and $B$ are at the same temperature, a noninteracting composite system $AB$ [formed from $A$ and $B$] is also at that temperature. Suppose that  $E_j$ is an energy level of system $A$ and $E_k$ is an energy level of $B$. 
Then, if there is no interaction between the two systems, $AB$ will have an energy level $E_j + E_k$. If we assume that for all systems prepared at temperature $T$, $\mathbb{P}(E_n) = \frac{1}{Z} f(E_n)$, then we have $f(E_j) f(E_k) Z_A Z_B = f(E_j+E_k) Z_{AB}$. But we have the freedom to adjust $f$ by an overall multiplicative constant, since the meaningful quantities are the probabilities and any prefactor will cancel when we divide by the partition function. 
So, we can declare $f(0) = 1$, which yields $Z_A Z_B = Z_{AB}$ and thus $f(E_j+E_k) = f(E_j) f(E_k)$. And this is just Cauchy’s functional equation for the exponential. So, provided that $f$ is continuous at even a single point, then $f(E) = e^{-\beta E}$, where the `coolness' $\beta$ labels the equivalence classes of thermal equilibrium.'' 
See \url{https://golem.ph.utexas.edu/category/2020/06/getting_to_the_bottom_of_noeth.html}, accessed 1 October 2020.
}
This derivation will implicitly motivate the construction of the effective temperature that culminates in \S \ref{sec:EffectiveTemperature}. While unlike more classical derivations ours does not motivate the introduction of entropy, the standard information-theoretic motivation provides a more than adequate remedy, and the Faddeev characterization of entropy is also a symmetry argument \cite{faddeev1956concept,baez2011characterization}. 

The key observation is that energy is only defined up to an additive constant, i.e., only energy differences are physically meaningful. This and the basic postulate that state probabilities are functions of state energies together imply that
\begin{equation}
\mathbb{P}(E_k) = \frac{f(E_k)}{\sum_j f(E_j)} = \frac{f(E_k + \varepsilon)}{\sum_j f(E_j + \varepsilon)}
\end{equation}
for some function $f$ and $\varepsilon$ arbitrary. Define
\begin{equation}
g_E(\varepsilon) := \frac{\sum_j f(E_j + \varepsilon)}{\sum_j f(E_j)}
\end{equation}
and note that $ g_E(0) = 1$ by definition. It follows that
\begin{equation}
\mathbb{P}(E_k) = \frac{f(E_k)}{\sum_j f(E_j + \varepsilon)} g_E(\varepsilon) = \frac{f(E_k + \varepsilon)}{\sum_j f(E_j + \varepsilon)}.
\end{equation}
Therefore $g_E(\varepsilon) f(E_k) = f(E_k + \varepsilon)$, implying that 
\begin{equation}
f(E_k + \varepsilon) - f(E_k) = (g_E(\varepsilon) - 1) \cdot f(E_k).
\end{equation} 

Since $ g_E(0) = 1$, we obtain $f'(E_k) = g'_E(0) f(E_k)$, and in turn
\begin{equation}
f(E_k) = C \exp (g'_E(0)E_k).
\end{equation}

\begin{svgraybox}
Without loss of generality, we can set $\beta := -g'_E(0)$ and $C \equiv 1$, which produces the Gibbs distribution so long as the temperature is \emph{defined} as $\beta^{-1}$. 
\end{svgraybox}

We note that the present derivation can be made rigorous (e.g., details involving continuity and the Cauchy functional equation) without substantial difficulty, but also without substantive additional insight. Also, $g_E(\varepsilon) = \exp(-\beta \varepsilon)$ so that $g_E \equiv g$, as required for the self-consistency of the argument. Although the present derivation is only appropriate for $\beta$ fixed, this just amounts to considering the canonical ensemble in the first place.

Finally, we reiterate that there are just a handful of symmetry and scaling principles collectively underlying the present derivation and that of the effective temperature below. In concert with the standard information-theoretical infrastructure for entropy, these principles provide an exceptionally parsimonious framework for the equilibrium statistical physics of finite systems.

\subsection{\label{sec:Descriptions} Statistical and physical system descriptions}

Consider now a stationary system with state space $[n]$. For our purposes, a sufficient \emph{statistical} description of such a system is provided by the $(n+1)$-tuple $(p_1,\dots,p_n,t_\infty) = (p,t_\infty)$, where $p_j := \mathbb{P}(s(t) = j)$ is the probability for the system to be in state $j \in [n]$, and where $t_\infty$ is a suitable characteristic or ``effective'' timescale. 
\footnote{
For technical reasons we will impose the nondegeneracy requirement $p_j > 0$ throughout our discussion. 
}
\footnote{
As we shall see, it turns out that physical considerations constrain $t_\infty$ to share many of the features of a mixing time or inverse energy gap (i.e., a relaxation time). 
}
Defining $t_j := t_\infty p_j$, the $n$-tuple $t := t_\infty p = (t_1,\dots,t_n)$ provides an alternative but completely equivalent description of the system, since the probability constraint $\sum_j p_j = 1$ implies that $t_\infty = \sum_j t_j$. We will use both of these descriptions interchangeably below without further comment. 

Meanwhile, a sufficient \emph{physical} description of the system is provided by the $(n+1)$-tuple $H := (E_1,\dots,E_n,\beta^{-1}) = (E,\beta^{-1})$, where $E_j$ is an effective energy for state $j$, and where $\beta$ is an effective inverse temperature. It will also be convenient to introduce $\gamma := \beta E$, noting that $\beta H = (\gamma,1)$.

Below, we will construct well-defined and essentially unique physically reasonable and mutually inverse maps (see Figure \ref{fig:t2H})
\begin{equation}
\label{eq:maps}
F_H(t) = H, \quad F_t(H) = t.
\end{equation}
The map $F_t$ extends the familiar Gibbs relation \eqref{eq:Gibbs}, and the relationship between $t_\infty$ and $\beta$ plays a pivotal role in the construction of both $F_H$ and $F_t$. In particular, we will determine $\beta$ as a function of $t$ in \eqref{eq:explicitFH1}, whereupon the Gibbs relation and equation \eqref{eq:E0} for the reference energy will complete the detailed specification of $F_H$.

\begin{figure}[htbp]
\includegraphics[trim = 10mm 20mm 10mm 25mm, clip, width=\textwidth,keepaspectratio]{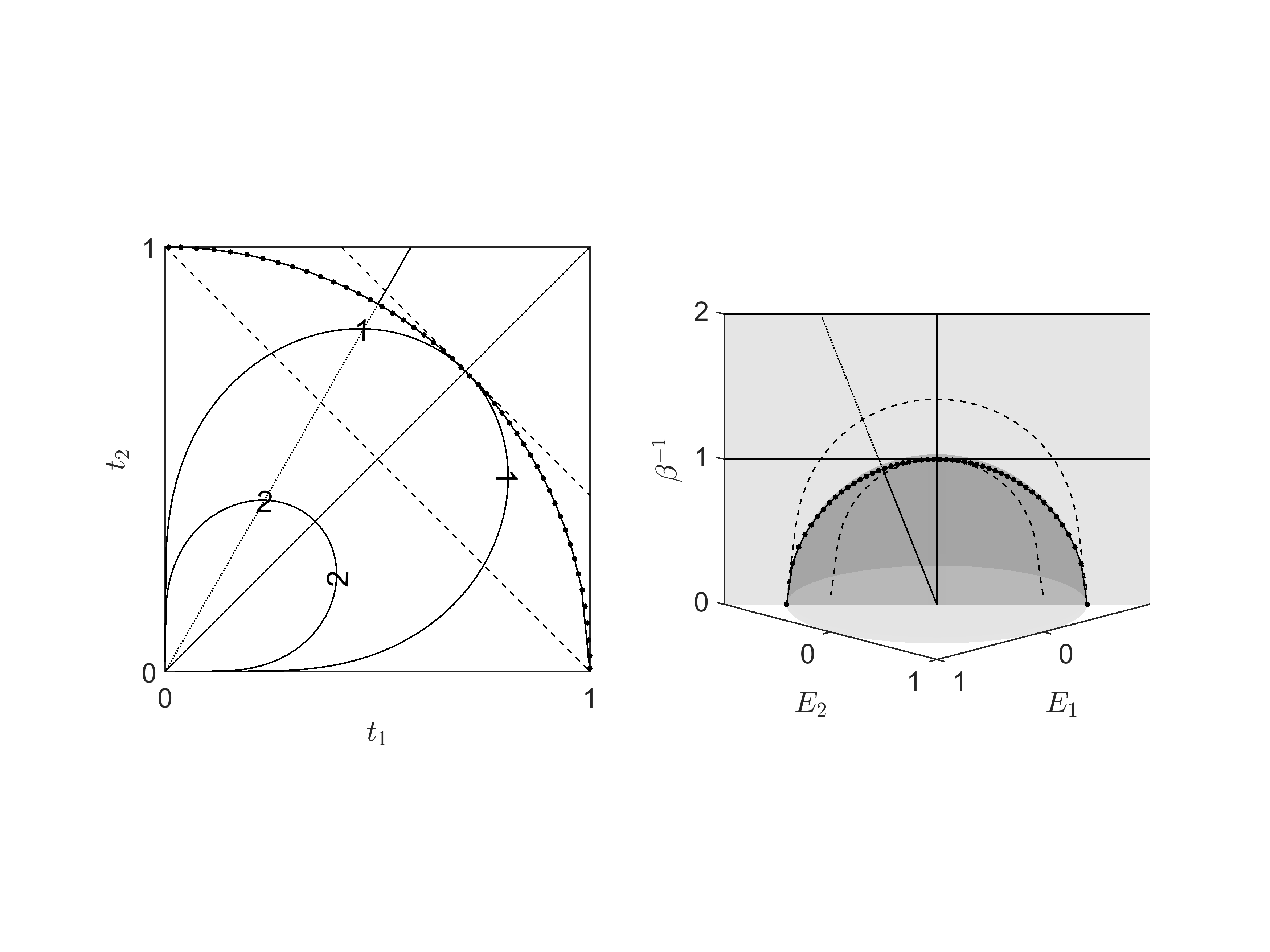}
\caption{ \label{fig:t2H} Geometry of the maps \eqref{eq:maps} for $n = 2$ states. Level curves of $\beta^{-1} = 1,2$ (solid contours) and of $t_\infty = 1, \sqrt{2}$ (dashed contours) are shown in both coordinate systems. The actions of the maps are also shown explicitly for circular arcs and rays.}
\end{figure} 

Because adding an arbitrary constant to the effective state energies merely amounts to a shift of a potential with no physical relevance, it is convenient to specify a reference energy, at least temporarily. With the preceding considerations in mind, and without any loss of generality, we impose the constraint
\footnote{
NB. This does not entail a specification of the internal energy (or any other physically meaningful quantity) \emph{\`{a} la} Jaynes \cite{jaynes1957information}.
}
\begin{equation}
\label{eq:E0}
\frac{1}{n} \sum_j E_j = 0.
\end{equation}
Note that we may later enforce any other convenient reference energy, e.g., $\min_j E_j \equiv 0$, $n^{-1} \sum_j E_j = \beta^{-1}$, etc.

\subsection{\label{sec:PrelliminaryAlgebra} Preliminary algebra}

For systems in thermal equilibrium, it is natural to require that $\beta$ is the inverse of the physical temperature, i.e., that the effective and physical temperatures coincide. In this case the fundamental principle of equilibrium statistical physics embodied by the Gibbs relation may be expressed as
\begin{equation}
\label{eq:Gibbs}
p_k = Z^{-1}e^{-\gamma_k}
\end{equation}
and regarded as a map $p = F_p(H)$. Here as usual $Z := \sum_j e^{-\gamma_j}$ is the partition function. 

By provisionally ignoring whether or not a generic stationary system is actually in thermal equilibrium, \eqref{eq:Gibbs} can be viewed as a constraint linking its physical and statistical descriptions. We will justify this interpretation below by using elementary symmetries and scaling relationships to specify (up to an overall constant) the inverse $F_H$ of an augmentation $F_t$ of the Gibbs map $F_p$. 

Taking logarithms on both sides of \eqref{eq:Gibbs} yields 
\begin{equation}
\label{eq:logp}
-\log Z - \gamma_k = \log p_k.
\end{equation} 
Meanwhile, the constraint \eqref{eq:E0} implies that $n^{-1} \sum_j \gamma_j = 0$. Combining this observation with arithmetic averaging of both sides of \eqref{eq:logp} leads to the result
\begin{equation}
\label{eq:logZ}
-\log Z = \frac{1}{n} \sum_j \log p_j.
\end{equation} 

Substituting \eqref{eq:logZ} into \eqref{eq:logp} and solving for $\gamma_k$ shows that
\begin{equation}
\label{eq:betaEk}
\gamma_k = \frac{1}{n}\sum_j \log p_j - \log p_k.
\end{equation}
Since $\beta F_H(t) = (\gamma,1)$,
\begin{equation}
\label{eq:normbetaH2}
\lVert \beta  F_H(t) \rVert^2 = \lVert \gamma \rVert^2 + 1,
\end{equation}
where $\lVert \cdot \rVert$ denotes the usual Euclidean norm. 
That is, $\lVert \beta  F_H(t) \rVert = \sqrt{\lVert \gamma \rVert^2 + 1}$ can be explicitly computed in terms of $\gamma$ (and by \eqref{eq:betaEk} also in terms of $p$) alone.

Therefore, in order to determine $\beta$, it remains chiefly to determine $\lVert F_H(t) \rVert$, since $\lVert \beta  F_H(t) \rVert$ is known from \eqref{eq:normbetaH2} and we tautologically have that
\begin{equation}
\label{eq:betaNormQuotient}
\beta = \lVert \beta  F_H(t) \rVert/\lVert F_H(t) \rVert.
\end{equation}
To determine $\lVert F_H(t) \rVert$, we will establish two results on scaling and geometry next.

\subsection{\label{sec:Scaling} A scaling result}

Dimensional considerations imply that if $\beta$ is determined by any well-behaved map $F_H$, then it must depend on some constant governing parameter $S$ in addition to $t$. That is, $\beta = f(t; S)$. By the Buckingham $\Pi$-theorem \cite{buckingham1914physically, barenblatt2003scaling}, $\beta = S^\xi t_\infty^\omega \Psi(p)$ for some $\xi$ and $\omega$, where $\Psi$ is dimensionless.

Consider for the moment a system governed by a Hamiltonian $\mathcal{H}(X,P)$. If $C$ is a constant, the transformation $t \mapsto t' := t/C$ induces the transformation $t_{\infty} \mapsto t'_{\infty} := t_{\infty}/C$ as well as the extended canonical (pure scale) transformation \cite{goldstein2001classical}
\begin{equation}
\label{eq:extendedcanonicaltransformation}
X \mapsto X' = X, \quad P \mapsto P' = CP, \quad \mathcal{H} \mapsto \mathcal{H}' = C\mathcal{H}.
\end{equation}

Since the transformation \eqref{eq:extendedcanonicaltransformation} can be considered as a change of units, it necessarily leaves the actual (vs. effective) Gibbs factor $e^{-\beta \mathcal{H}}$ invariant. That is, $\beta \mathcal{H} = \beta' \mathcal{H}' = \beta' C \mathcal{H}$. This observation immediately yields that $\beta \mapsto \beta' = \beta/C$. Physical consistency therefore demands that 
\begin{equation}
\label{eq:abstractscaling}
F_H(t/C) = C \cdot F_H(t).
\end{equation}
From this, it follows that $\omega = 1$, so without loss of generality
\begin{equation}
\label{eq:scaling}
\beta = S^{-1} t_\infty \Psi(p),
\end{equation} 
where the constant $S$ carries units of action (say, $S = \hbar$).

\subsubsection{\label{sec:ScalingArguments}Additional arguments in support of \eqref{eq:abstractscaling} and \eqref{eq:scaling}}

A reader fully convinced by the argument just above can safely skip this section.

\paragraph{\label{sec:idealgasscaling} Ideal gas systems}

Consider a \emph{Gedankenexperiment} with two systems, comprised respectively of finite ideal gas samples with particle masses $m$ and $m' = m/C$, each in identical freefalling containers in contact with isotropic thermal baths, and with the same initial conditions in phase space. Let $P = P'$ denote the common rms momentum of particles in both systems: the respective inverse temperatures of the two systems are then in common proportion to $m/P^2$ and $m'/P'^2 = (m/C)/P^2$. 

Insofar as the system microstates are not of interest in equilibrium, the systems may be respectively described by, e.g. the quintuples $(m, v, P, t_\infty, \beta)$ and $(m', v', P', t_\infty', \beta') = (m/C, Cv, P, t_\infty/C, \beta/C)$, where $v^{(\cdot)}$ denotes a rms velocity and here $t_\infty^{(\cdot)}$ denotes \emph{any} characteristic timescale of the same nature in both systems.

Both systems follow the same trajectory through phase space, albeit at rates that differ by constant factors, and we see that $\beta$ scales as $t_\infty$ for ideal gases, and hence (by coupling with an ideal gas bath) for general systems also. Therefore, consistency with elementary equilibrium statistical physics requires that $\beta$ also scales as $t_\infty$.

\paragraph{\label{sec:ckms}The classical KMS condition}

Another argument along similar lines to that in \S \ref{sec:Scaling} for the scaling behavior of $\beta$ w/r/t $t_\infty$ directly invokes the classical Kubo-Martin-Schwinger (KMS) condition. To begin, we recall the usual (quantum) KMS condition before formally deriving its classical analogue in the limit $\hbar \rightarrow 0$ by way of background.

A quantum Hamiltonian $\mathcal{\hat H}$ has thermal density matrix
\begin{equation}
\label{eq:thermaldensitymatrix}
\hat \rho := Z^{-1}e^{-\beta \mathcal{\hat H}},
\end{equation}
where $Z := \mbox{Tr}(e^{-\beta \mathcal{\hat H}})$, and the time evolution of an observable $\hat A$ in the Heisenberg picture is given as usual by $\tau_t(\hat A) := e^{i\mathcal{\hat H}t/\hbar} \hat A e^{-i\mathcal{\hat H}t/\hbar}$. 

The quantum Gibbs rule $\langle \hat A \rangle = \mbox{Tr}(\hat \rho \hat A)$, with $\hat \rho$ given by \eqref{eq:thermaldensitymatrix}, is generalized by the KMS condition \cite{gallavotti1975classical,parisi1998statistical}
\begin{equation}
\label{eq:kms}
\left \langle \tau_t(\hat A) \hat B \right \rangle = \left \langle \hat B\tau_{t+i\hbar\beta}(\hat A) \right \rangle.
\end{equation}
For convenience, we recall a formal derivation of \eqref{eq:kms} from the Gibbs rule and the cyclic property of the trace:
\begin{eqnarray}
\left \langle \tau_t(\hat A) \hat B \right \rangle & = & Z^{-1}\mbox{Tr}(e^{-\beta \mathcal{\hat H}} e^{i\mathcal{\hat H}t/\hbar} \hat A e^{-i\mathcal{\hat H}t/\hbar} \hat B) \nonumber \\
& = & Z^{-1}\mbox{Tr}(\hat B e^{i\mathcal{\hat H}z/\hbar} \hat A e^{-i\mathcal{\hat H}t/\hbar}) \nonumber \\
& = & Z^{-1}\mbox{Tr}(\hat B e^{i\mathcal{\hat H}z/\hbar} \hat A e^{-i\mathcal{\hat H}z/\hbar} e^{-\beta \mathcal{\hat H}}) \nonumber \\
& = & \left \langle \hat B\tau_z(\hat A) \right \rangle \nonumber
\end{eqnarray}
where here we have written $z := t+i\hbar\beta$.

Following \cite{gallavotti1975classical}, we have by \eqref{eq:kms} the following precursor to the classical KMS condition:
\begin{equation}
\label{eq:underformedclassicalkms}
\left \langle \frac{[\tau_t(\hat A),\hat B]}{i\hbar} \right \rangle = \left \langle \hat B \left( \frac{\tau_z(\hat A) - \tau_t(\hat A)}{i\hbar} \right) \right \rangle.
\end{equation} 

Recall that as $\hbar \rightarrow 0$, $\tau_t(\hat A)$, $\hat B$ and $[\tau_t(\hat A),\hat B]/i\hbar$ respectively correspond to or ``undeform'' into classical analogues $A$, $B$ and $\{A,B\}$, where $A$ has an implicit time dependence (i.e., $\partial_t A = 0 \not \equiv dA/dt$) and $B$ does not (i.e., $B$ is evaluated at $t=0$). 

Now (via an implicit assumption about the analyticity of $\tau_z$ which forms the actual substance of the KMS condition) we have that
\begin{equation}
\lim_{\hbar \rightarrow 0} \frac{\tau_z(\hat A) - \tau_t(\hat A)}{i\hbar} = \beta \frac{dA}{dt} = \beta \{A, \mathcal{H} \}
\end{equation}
where $\mathcal{H}(X,P)$ is the classical Hamiltonian. Therefore in the limit $\hbar \rightarrow 0$, \eqref{eq:underformedclassicalkms} formally becomes the classical KMS condition (see also \cite{parisi1998statistical})
\begin{equation}
\label{eq:classicalkms}
\left \langle \{A,B\} \right \rangle = \beta \left \langle B \{A,\mathcal{H}\} \right \rangle.
\end{equation}

As in \S \ref{sec:idealgasscaling}, here let $t_\infty^{(\cdot)}$ denote \emph{any} characteristic timescale of the system. Dilating the dynamical rate by a constant factor $C$ has the effect that $t_{\infty} \mapsto t'_{\infty} = t_{\infty}/C$ and also induces the extended canonical (pure scale) transformation \eqref{eq:extendedcanonicaltransformation}. It follows that $\partial_X = \partial_{X'}$ and $\partial_P = C\partial_{P'}$, whence $\{A,B\} = C\{A,B\}'$ and $\{A,\mathcal{H}\} = C\{A,C^{-1}\mathcal{H}\}' = \{A, \mathcal{H}'\}'$ (here $\{\cdot,\cdot\}'$ denotes the Poisson bracket w/r/t $(X',P')$). Along with \eqref{eq:classicalkms}, this in turn gives that  
\begin{equation}
\beta = \frac{\left \langle \{A,B\} \right \rangle}{ \left \langle B \{A,\mathcal{H}\} \right \rangle} = \frac{\left \langle C\{A,B\}' \right \rangle}{ \left \langle B \{A,\mathcal{H}'\}' \right \rangle} = C \beta'.
\end{equation}

Therefore $\beta' =\beta/C$ and we see once more that $\beta$ scales as any characteristic time $t_\infty$. Again, consistency with traditional equilibrium statistical physics dictates that an effective inverse temperature should also scale as $t_\infty$.

\paragraph{\label{sec:tth} Thermal time hypothesis}

The one-parameter modular group of $\hat \rho$ (as defined in \eqref{eq:thermaldensitymatrix}) that appears in the Tomita-Takesaki theory of von Neumann algebras \cite{bratteli2012operator} can be shown to coincide with the time evolution group \cite{connes1994neumann}: if $s$ is the modular parameter and $t$ is the physical time, then
\begin{equation}
\label{eq:modular}
t = \hbar \beta s.
\end{equation}
In particular, $s$ does not depend on $\beta$. \footnote{While time evolution for von Neumann algebras is only of direct interest in the infinite-dimensional setting, its significance for the present context is nevertheless readily apparent.} 

The \emph{thermal time hypothesis} (TTH) articulated by Connes and Rovelli \cite{connes1994neumann} (see also \cite{martinetti2003diamond, rovelli1993statistical, tian2005sitter, rovelli2011thermal}) states that physical time is determined by the modular group, which is in turn determined by the state. 

Besides implying Hamiltonian mechanics, the TTH simultaneously inverts and generalizes the KMS condition (see \S \ref{sec:ckms}) and hence also the Gibbs relation \eqref{eq:Gibbs}, with temperature providing the physical link between time evolution and equilibria. But its key implication here is \eqref{eq:modular}, by which $\beta$ scales as any characteristic time $t_\infty$; as before consistency demands the same scaling behavior for an effective inverse temperature.

\paragraph{\label{sec:counterscaling}Counterarguments for alternative scaling behavior}

Despite the scaling arguments presented above, we might nevertheless feel compelled to consider alternative scaling behavior, with an effective inverse temperature of the form $\|t\|^\omega \sqrt{\|\gamma\|^2 +1}$. However, for $\omega \ne 1$ this quantity does not converge in a natural way for archetypal Anosov systems (see \S \ref{sec:Anosov}), nor by extension does it appear to be relevant to the example of a two-dimensional ideal gas. Furthermore, its physical relevance for a single Glauber-Ising spin (see \S \ref{sec:TwoStateSystems}) is dubious for $\omega \ne 1$. Such behavior can be viewed as providing additional (albeit more circumstantial) evidence for an effective inverse temperature scaling as $t_\infty$, as can the validity of the \emph{Ansatz} suggested by this scaling behavior for synchronization frequencies of Kuramoto oscillators (see \ref{sec:Synchronization}).

\subsection{\label{sec:Geometry} A geometry result}

The transformation $t \mapsto t' := t/C$ leaves $p$ invariant. Meanwhile, $\gamma_k$ depends only on $p$, so both $\gamma$ and $\beta H = (\gamma,1)$ are also invariant under this transformation, in accordance with \eqref{eq:abstractscaling}. In other words, $p$ is positive homogeneous of degree zero in both $t$ and $H$, i.e., $F_p(Ct) := t/t_\infty = F_p(t)$ and $F_p(CH) = F_p(H)$.
\footnote{
Recall that a function $f$ defined on a cone in $\mathbb{R}^n \backslash \{0\}$ is said to be \emph{positive homogeneous of degree $a$} iff $f(C {\bm x}) = C^a f({\bm x})$ generically for $C > 0$.
}
\footnote{
Yet another equivalent characterization is that $p$ is constant (away from the origin) on rays through the origin of the form $\mathbb{R}t$ and $\mathbb{R}H$.
}

Recall that Euler's homogeneous function theorem states that if $f \in C^1(\mathbb{R}_+^n)$ is positive homogeneous of degree $a$, then $\langle x, \nabla_x f(x) \rangle = a \cdot f(x)$ \cite{reiss1997methods}. Since each component $p_k$ of $p$ is positive homogeneous of degree zero as a function of both $t$ and $H$, it follows that $\langle \nabla_t p_k, t \rangle = 0 = \langle \nabla_H p_k, H \rangle$. Therefore each of the gradients $\nabla_t p_k$ and $\nabla_H p_k$ are tangent to spheres centered at the origin of their respective coordinate systems. 

Furthermore, the gradients $\nabla_t p_k$ and $\nabla_H p_k$ are nondegenerate: an explicit calculation shows that $\partial p_k/\partial E_j = \beta p_k (-\delta_{jk}+p_j) \ne 0$, and $\partial p_k / \partial (\beta^{-1}) = \beta^2 p_k (E_k-U)$, where as usual $U := \sum_j p_j E_j$. Meanwhile, $\partial p_k / \partial t_j = (\delta_{jk}t_\infty - t_k)/t_\infty^2 \ne 0$. Taking appropriate directional derivatives makes it easy to see that the constraint \eqref{eq:E0} does not affect the nondegeneracy of these gradients.

Consider now the unique decomposition of a vector differential as $dv = dv^\parallel + dv^\perp$, where the terms on the right hand side are respectively parallel and perpendicular to $v$. That is, $dv^\parallel := (\langle dv, v \rangle / \langle v, v \rangle) v$, and $dv^\perp := v - dv^\parallel$.  It is easy to see that $dt^\perp = 0 \iff dp = 0 \iff dH^\perp = 0$ from the preceding considerations. Moreover, $\langle \nabla_{t^\perp} p_k, dt^\perp \rangle = \langle \nabla_t p_k, dt \rangle = dp_k = \langle \nabla_H p_k, dH \rangle = \langle \nabla_{H^\perp} p_k, dH^\perp \rangle$. That is, the nondegenerate integral curves of both gradient flows are arcs on spheres centered at the origin. Since a smooth change of coordinates maps integral curves into integral curves \cite{choquet1977analysis}, it follows that the respective spheres on which these arcs lie must also map to each other under any smooth maps $F_H$ and $F_t$ satisfying \eqref{eq:E0} and \eqref{eq:Gibbs}. We therefore have

\begin{lemma} 
A well-behaved map $F_H$ that respects \eqref{eq:E0} and \eqref{eq:Gibbs} sends rays and sphere orthants centered at the origin to rays and hemispheres centered at the origin, respectively. In particular, a well-behaved map $F_H$ that respects \eqref{eq:E0} and \eqref{eq:Gibbs} must satisfy
\begin{equation}
\label{eq:geometry}
\lVert {\bm s} \rVert = \lVert t \rVert \Rightarrow \lVert F_H({\bm s}) \rVert = \lVert F_H(t) \rVert.
\end{equation}
\end{lemma}

\subsection{\label{sec:EffectiveTemperature} The effective temperature}

Let $u_j := \lVert t \rVert/\sqrt{n}$, so that $\lVert u \rVert \equiv \lVert t \rVert$ and $F_H(u) \equiv (0, \dots, 0, \beta_u^{-1})$. Now $\lVert F_H(u) \rVert^2 = \beta_u^{-2}$, and by \eqref{eq:geometry}
\begin{equation}
\lVert F_H(t) \rVert^2 = \lVert F_H(u) \rVert^2 = \beta_u^{-2}.
\end{equation}
Therefore, by \eqref{eq:normbetaH2} and \eqref{eq:betaNormQuotient},
\begin{equation}
\beta^2 = \frac{\lVert \beta F_H(t) \rVert^2}{\lVert F_H(t) \rVert^2} = \frac{\lVert \gamma \rVert^2 + 1}{\beta_u^{-2}}.
\end{equation}
Taking square roots of the far left- and right-hand sides yields
\begin{equation}
\label{eq:prelimbeta}
\beta = \beta_u \sqrt{\lVert \gamma \rVert^2 + 1}.
\end{equation}

Meanwhile, \eqref{eq:scaling} implies that 
\begin{equation}
\label{eq:KMSUT}
\beta_u = S^{-1} \lVert u \rVert = S^{-1} \lVert t \rVert = S^{-1} t_\infty \lVert p \rVert,
\end{equation}
where $S$ is a fixed constant with units of action (say, $S = \hbar = 1$). 
To see the first equality of \eqref{eq:KMSUT}, note that $u_j := \lVert t \rVert/\sqrt{n}$ implies that $\sum_j u_j = \sqrt{n} \lVert t \rVert = \sqrt{n} \lVert u \rVert \equiv u_\infty$. Since $u_j/u_\infty = n^{-1}$, it follows that $\Psi(u/u_\infty) =: \psi(n)$ is a function of $n$ alone. Now \eqref{eq:scaling} gives that $\beta_u = S^{-1} u_\infty \psi(n) = S^{-1} \psi(n) \sqrt{n} \lVert u \rVert$. Without loss of generality, the term $\psi(n) \sqrt{n}$ can be absorbed into the constant $S$.
\footnote{
In \S VIII of \cite{huntsman2010anosov} we discuss detailed evidence that physical consistency appears to demand $\psi(n) = 1/\sqrt{n}$, as this choice (somewhat counterintuitively) appears to be the unique one giving a well-defined limit in the microcanonical ensemble for discretizations of two physically paradigmatic systems.
}

Combining \eqref{eq:prelimbeta} and \eqref{eq:KMSUT} with $x = 1$ therefore yields
\begin{equation}
\label{eq:temperature}
\beta = t_{\infty} \lVert p \rVert \cdot \sqrt{\lVert \gamma \rVert^2 +1},
\end{equation}
whereupon \eqref{eq:betaEk} leads to explicit expressions for $F_H$:
\begin{align}
\label{eq:explicitFH1}
\beta^{-1} & = \frac{1}{t_\infty \lVert p \rVert} \left( \sum_{k=1}^n \left[ \frac{1}{n} \sum_{j=1}^n \log \frac{p_j}{p_k} \right]^2 +1 \right)^{-1/2}; \\
\label{eq:explicitFH2}
E_k & = \beta^{-1} \cdot \frac{1}{n} \sum_{j=1}^n \log \frac{p_j}{p_k}.
\end{align}
Similarly, $F_t$ is given explicitly (after shifting so that \eqref{eq:E0} is satisfied) as
\begin{align}
\label{eq:explicitFt1}
p_k & = Z^{-1} e^{-\beta E_k}; \\
\label{eq:explicitFt2}
t_\infty & = \lVert p \rVert^{-1} \cdot \left( \lVert E \rVert^2 + \beta^{-2} \right)^{-1/2}.
\end{align}

\begin{svgraybox}
To review, the derivation of the (Gibbs distribution and the) effective temperature rested on two basic symmetry assumptions and two derived symmetries. The basic assumptions are that
\begin{itemize}
	\item the zero point of energy is physically irrelevant;
	\item the probability of a state depends only on its energy.
\end{itemize}
The derived symmetries are that
\begin{itemize}
	\item changing units of time leaves $\beta \mathcal{H}$ invariant;
	\item any physically nice bijection $t \leftrightarrow H$ preserves rays and radii.
\end{itemize}
\end{svgraybox}

\subsection{\label{sec:ProductSystems}Product systems, the ideal gas, and implications for $t_\infty$}

Perhaps the most fundamental property of the ordinary temperature is intensivity. Imposing a few simple physical requirements such as the intensivity of the effective temperature $\beta^{-1}$ for simple product systems (which is a symmetry requirement in keeping with our overall theme) turns out to place significant physical constraints on the functional form of reasonable candidates for the timescale $t_\infty$, as we shall illustrate below. It seems likely that imposing similar requirements for (subsystems of) closed interacting systems such as coupled map lattices \cite{chazottes2005dynamics} will at least mirror--and probably augment--constraints of the sort discussed here, but analyses of interacting systems will almost surely be much more technically challenging.


\subsubsection{\label{sec:BasicProductResults}Basic results for product systems}

Consider $N$ systems sharing a common probability measure $p$ on $[n] = \{1,\dots,n\}$. Writing 
\begin{equation}
\label{eq:b0}
b := (\beta/t_\infty)^2 \equiv \lVert p \rVert^2 \cdot (\lVert \gamma \rVert^2 + 1) \nonumber
\end{equation} 
for convenience and using a superscript $\otimes$ to indicate the product system, it can be shown that
\begin{equation}
\label{eq:botimes0}
b^\otimes = Nn^{N-1} \lVert p \rVert^{2(N-1)} \cdot \left ( \lVert p \rVert^2 \left [ \lVert \gamma \rVert^2 + \{Nn^{N-1}\}^{-1} \right ] \right ).
\end{equation}

The somewhat peculiar way of writing the right hand side of \eqref{eq:botimes0} is motivated by the fact that in the limit of large $\lVert \gamma \rVert^2$, the term in parentheses tends to $b$, in which event
\begin{equation}
\label{eq:approximatebotimesuniform0}
b^\otimes \approx N n^{N-1} \lVert p \rVert^{2(N-1)} \cdot b = N \frac{n^\otimes \| p^\otimes \|^2}{n \| p \|^2}\cdot b.
\end{equation}

Recall that the harmonic mean $\langle f \rangle_h$ of a function $f$ on $[N]$ is given by 
\begin{equation}
\langle f \rangle_h^{-1} := \langle 1/f \rangle_a \equiv N^{-1} \sum_m f_m^{-1}, \nonumber
\end{equation}
where $\langle \cdot \rangle_a$ indicates the arithmetic mean. If (in the present context of a collection of subsystems with identical measures) we make the physically reasonable stipulation of intensivity for the effective temperature (cf. \S \ref{sec:Synchronization}), i.e. $\beta^\otimes = \langle \beta \rangle_h$, then since $\beta = \sqrt{b} t_\infty$ we must have that 
\begin{equation}
\label{eq:approximatetinftyotimesuniform}
t_\infty^\otimes = \sqrt{b/b^\otimes} \cdot \langle t_\infty \rangle_h.
\end{equation} 

If furthermore the number $n$ of states in each component system tends to infinity while $p \equiv p_{(n)}$ remains sufficiently uniform, the intensivity property \eqref{eq:approximatetinftyotimesuniform} turns out to take the form
\begin{equation}
\label{eq:aggregatetinfty}
t_\infty^\otimes = N^{-1/2} \langle t_\infty \rangle_h.
\end{equation}

\subsubsection{\label{sec:IdealGas}The two-dimensional ideal gas on a compact surface of constant negative curvature}

An example of particular interest along the lines above is furnished by the geodesic flow on a compact surface of constant negative curvature (see \S \ref{sec:FreeParticle}). In this context, \eqref{eq:aggregatetinfty} gives a recipe for applying our framework to the ideal gas (with or without a thermostat).

Besides the apparently well-defined value of $\beta$ for the geodesic flow (i.e., a single particle) on a compact surface of constant negative curvature, the essential observation for establishing the plausible consistency of $\beta^\otimes$ with the physical inverse temperature is simply one of scaling behavior. We detail this here.

It was shown in \cite{collet1984perturbations} that the $L^2$ mixing time of the geodesic flow is 1/2 for reasonably well-behaved observables. Taking this (or with trivial modifications, any other constant timescale, e.g. the genus-independent inverse topological entropy [see section \S \ref{sec:L2markov}]) as $t_\infty$ for a single flow with speed $v=1$, we have that $vt_\infty = 1/2$ more generally. Now $\langle t_\infty \rangle_h = \langle t_\infty^{-1} \rangle_a^{-1} = 1/2\langle v \rangle_a$. For a two-dimensional ideal gas $\langle v \rangle_a = \sqrt{\pi/2\beta m}$, so for $\beta^\otimes$ to equal the physical inverse temperature we must have by \eqref{eq:aggregatetinfty} that
\begin{equation}
\label{eq:tinfty2Didealgas}
t_\infty^\otimes = \frac{1}{2 \langle v \rangle_a \sqrt{N}} = \sqrt{\frac{\beta m}{2\pi N}}.
\end{equation}

The quadratic dependence on $\beta$ (and on $m$) in the above equation has a simple explanation consistent with $\beta$ scaling as $t_\infty$. While the argument that $\beta$ scales as $t_\infty$ ceases to apply when we only vary $v$, it \emph{does} apply when we hold a phase space trajectory fixed, and in this event $\beta$, $m$ and $t_\infty$ all scale identically (cf. \S \ref{sec:idealgasscaling}). Indeed, in the single-particle case $\beta \equiv 2/mv^2$ and $vt_\infty = 1/2$, so $\beta = 8t_\infty^2/m$.  

Consequently $\beta^\otimes$ and the physical inverse temperature scale identically: in particular, both are constant in the limit of large $N$, and incorporating an appropriate constant into the definition of $\beta$ yields equality (cf. \S \ref{sec:OverallScale}).

It is worth noting here that naive discretizations of an ideal gas with obvious configuration space geometry, boundary conditions, ultraviolet cutoffs, etc. do \emph{not} exhibit reasonable scaling limits, a fact which motivated our analysis of the rather esoteric version and context of the ideal gas considered here.

\subsubsection{\label{sec:TensorMarkov}Products of Markov processes and constraints on $t_\infty$}



The detailed behavior of the relationship \eqref{eq:aggregatetinfty} allows us to rule out a number of potential candidates for a broadly applicable $t_\infty$. 

For instance, recurrence, hitting, covering or similar timescales do not appear to be suitable candidates. Additionally, quantities such as the recurrence rates of a flow \cite{saussol2009introduction} or a so-called cutoff for a family of product Markov processes \cite{barrera2006cut} are not appropriate choices in the present context simply because they do not have the necessary parametric dependence. 

While the form of \eqref{eq:approximatetinftyotimesuniform} and \eqref{eq:aggregatetinfty} suggest that the choice for $t_\infty$ should bear some qualitative similarily to a relaxation time \cite{schwarz1968kinetic}, we can also rule out a naive identification of $t_\infty$ with an inverse spectral gap in the context of Markovian dynamics, as we proceed to sketch.

For $m \in [N]$, let $Q^{(m)}$ be the generator of a (well-behaved) continuous-time Markov process on $[n_m]$. The composite Markov generator corresponding to evolving each of the $N$ processes simultaneously turns out to be 
\begin{equation}
\label{eq:tensorsum}
Q^\otimes = \sum_m I^{\otimes(m-1)} \otimes Q^{(m)} \otimes I^{\otimes(N-m)}.
\end{equation}
It is easily seen that the spectral gap of $Q^\otimes$ is just the smallest of the spectral gaps of the $Q^{(m)}$. In particular, if (as we shall assume henceforth)
\begin{equation}
Q^{(m)} = c_m Q
\end{equation}
for $c_m > 0$, then the spectral gap of $Q^\otimes$ is the product of the gap for $Q$ and $\min_m c_m$. This precludes a relation of the form \eqref{eq:approximatetinftyotimesuniform} or \eqref{eq:aggregatetinfty} for an inverse spectral gap and suggests that such a quantity is not a generically suitable choice for $t_\infty$. That said, a ``modified'' $L^2$ mixing time is related to an inverse spectral gap and does appear to be a viable generic candidate for $t_\infty$ (as does the similarly normalized inverse topological entropy: see \S \ref{sec:L2markov}), as we shall see below. For a reversible Markov process without product structure, this timescale and the inverse spectral gap coincide, and for the example of a single Glauber-Ising spin both equal $(2a)^{-1}$. The \emph{Ansatz} $t_\infty = (2a)^{-1}$ discussed in \S \ref{sec:TwoStateSystems} thus amounts roughly to (quite reasonably) equating the spectral gap of the generator and the dominant energy scale. 


While we dwell on the potential for a broadly applicable recipe for appropriately determining $t_\infty$, we must also consider the possibility (discussed in \S \ref{sec:OverallScale}) that no completely universal recipe exists. That is, it may be that appropriate choices for $t_\infty$ are necessarily context-dependent, for example in the same way that the Gibbs paradox illustrates that the entropy of a system can depend on the level of specification \cite{jaynes1992gibbs}. Indeed, detailed consideration of a classical Bose gas (not included here) suggests indicates that distinguishability of particles should inform the effective temperature if $t_\infty$ is given by a modified $L^2$ mixing time.

In any event, the proper specification of $t_\infty$ is clearly a central component of our effective framework for statistical physics, and the degree of universality with which this specification can be accomplished will be directly related to its ultimate physical significance. Nevertheless, as both the analogy with the Gibbs paradox and the characterization of individual systems varying in time or over some parametric ensemble show, even a context-dependent quantity can still have substantial physical relevance.

\subsubsection{\label{sec:L2markov}$L^2$ convergence of Markov processes and a modified mixing time}

As a preliminary to discussing the modified $L^2$ mixing time mentioned above, we first review here the ordinary $L^2$ mixing time for Markov processes. 
\footnote{
NB. We follow the standard convention in physics and dynamical systems theory for ``the'' $L^2$ mixing time, which differs somewhat from the mixing time function typically considered by probabilists.
} 
Given a (not necessarily reversible but well-behaved) Markov generator $Q$ with invariant distribution $p$, the corresponding \emph{Dirichlet form} is 
\begin{equation}
\label{eq:dirichlet}
\mathcal{E}(f) := \frac{1}{2} \sum_{j,k} p_j Q_{jk} (f_j-f_k)^2.
\end{equation}

Write 
\begin{equation}
\label{eq:lambdastar}
\lambda_* := \inf_{f} 2 \frac{\mathcal{E}(f)}{\mbox{Var}_{p}(f)},
\end{equation} 
where the infimum is over $f$ s.t. $\mbox{Var}_{p}(f) \ne 0$. It can be shown that $\lambda_*$ determines the $L^2$ convergence of the Markov process to stationarity: viz., $\lambda_*^{-1}$ is the $L^2$ mixing time. Furthermore, if $Q$ is reversible, its eigenvectors form a basis and $\lambda_*$ is the spectral gap.


For a product system of the form \eqref{eq:tensorsum} with $Q^{(m)} = c_m Q$, it can be shown that the infimum in \eqref{eq:lambdastar} is degenerate in the sense that its consideration amounts to ignoring various factors of the product. The nondegenerate minimum is (continuing an obvious notational convention)
\begin{equation}
\label{eq:lambdaotimes}
\lambda^\otimes := N \langle c \rangle_a \lambda_*,
\end{equation}
where $\lambda_*$ corresponds to $Q$.

Writing $\tau_\infty^{\otimes} := (\lambda^\otimes)^{-1}$ and $\tau_\infty^{(m)} := (c_m \lambda_*)^{-1}$, \eqref{eq:lambdaotimes} becomes 
\begin{equation}
\label{eq:L2tau}
\tau_\infty^\otimes = \langle \tau_\infty \rangle_h / N
\end{equation}
which differs from \eqref{eq:aggregatetinfty} only by a factor of $N^{-1/2}$ (though the context here is more general, as $p$ need not be close to uniform). 

A corresponding normalization of $\tau^\otimes_\infty$ that takes any product structure into account therefore appears to be a plausible general-purpose candidate for $t^\otimes_\infty$ satisfying \eqref{eq:approximatetinftyotimesuniform} in physically relevant cases. This modified $L^2$ mixing time is more physically natural than the usual $L^2$ mixing time because it measures the convergence of all the component processes, not just a single distinguished component process. It is properly normalized and avoids any degeneracies introduced by the tensor product structure. 

A similar result applies for the inverse topological entropy of a product system. Recall that the topological entropy of a system describes the rate at which the number of periodic orbits grows as a function of the orbital period. For this reason its inverse is a natural characteristic timescale, and it turns out that a straightforward normalization obeys \eqref{eq:aggregatetinfty}.

Indeed, the topological entropy of a product flow of the form $\phi^\otimes_t := \prod_m \phi^{(m)}_t$ with $\phi^{(m)}_t := \phi_{c_m t}$ satisfies $h(\phi^\otimes) = \sum_m c_m \cdot h(\phi)$ \cite{katok1997introduction}. So if we set $\tau^{(\cdot)}_\infty := 1/h(\phi^{(\cdot)})$, then we obtain a relation of precisely the form \eqref{eq:L2tau}. That is, the inverse topological entropy of a flow satisfies the same sort of product relationship as the modified $L^2$ mixing time. \footnote{In fact the inverse topological entropy and a topological (non-$L^2$) mixing time are related: see e.g. \cite{richeson2008chain}.} However, we focus on the mixing time as it may be more broadly applicable.

\subsection{\label{sec:ExamplesAndApplications}Elementary examples and applications}

We sketch some elementary examples and applications here. The application to Anosov systems and the chaotic hypothesis in \S \ref{sec:Anosov} is sufficiently involved and significant to demand special treatment, though it also informs an application to a two-dimensional ideal gas (see section \S \ref{sec:IdealGas}). Likewise, a thermodynamical analysis of the degradation of discrete memoryless channels is currently underway but not sketched here.

The framework presented here has been utilized for the characterization of computer network traffic \cite{huntsman2009effective} (another effort in a similar spirit is \cite{burgess2000thermal}). Although the potential scope of this framework appears to be quite broad, the key practical difficulties in applications are the identification of an appropriate state space (or discretization scheme) and characteristic timescale. The examples that we have thus far been able to identify all have complicating features in at least one of these regards. Nontrivial spin models, which might appear at first to give an ideal setting for exploring the effective temperature in detail, are deceptively difficult to deal with in this framework because of the subtle nature of timescales in glassy systems. \footnote{See \S XVIII of \cite{huntsman2010anosov} for a detailed discussion of this topic.} That said, a single Glauber-Ising spin will serve as an illustrative example in section \S \ref{sec:TwoStateSystems}.

For characterization of generic data sets (e.g., computer network traffic) the state space selection issues are similar to those confronted in the application of entropy methods, and the characteristic timescale may be dictated either by the data itself or by the collection interval. In many ways this ``descriptive thermodynamics'' is the simplest sort of application \cite{ford2006descriptive}, and in fact it motivated the present framework.

\subsubsection{\label{sec:TwoStateSystems}Two-state systems; a single Glauber-Ising spin}

Consider the simplest case of a two-state system as illustrated in Figure \ref{fig:t2H}. In this case we have that $\gamma_1 = \frac{1}{2}\log \frac{p_2}{p_1}$ and $\gamma_2 = \frac{1}{2}\log \frac{p_1}{p_2} = -\gamma_1$. Therefore trivial substitutions yield 
\begin{equation}
\label{eq:twostatebeta}
\beta = t_\infty \left [ (p_1^2 + p_2^2) \left( \frac{1}{2} \log^2 \frac{p_1}{p_2} + 1 \right) \right]^{1/2};
\end{equation}
\begin{equation}
\label{eq:twostateE}
E_1 = -\frac{1}{2\beta} \log \frac{p_1}{p_2}; \quad E_2 = -E_1.
\end{equation}

Going in the other direction, we first take $E_j \mapsto E_j - (E_1 + E_2)/2$ in accordance with \eqref{eq:E0}, so that again $E_2 = -E_1$ and $\gamma_2 = -\gamma_1$. Therefore
\begin{equation}
\label{eq:twostatep}
p_1 = \frac{e^{-\gamma_1}}{e^{-\gamma_1}+e^{\gamma_1}}; \quad p_2 = \frac{e^{\gamma_1}}{e^{-\gamma_1}+e^{\gamma_1}}.
\end{equation}
Moreover, $Z = 2 \cosh \gamma_1$, $\lVert p \rVert^2 = (1+\tanh^2 \gamma_1)/2$, and $\lVert \gamma \rVert^2 = 2\gamma_1^2$, from which it follows that
\begin{equation}
\label{eq:twostatetinfty}
t_\infty = \beta \left [ \frac{1 + \tanh^2 \gamma_1}{2} \left ( 2\gamma_1^2 + 1 \right ) \right ]^{-1/2}.
\end{equation}

As a physical incarnation of this example, consider the requirement that $\beta$ equal the \emph{actual physical inverse temperature} for a single Glauber-Ising spin $\sigma$ in a magnetic field. The spin dynamics are determined by an overall (spin flip) rate $a$ and $b := \tanh (\beta \mu h)$, where $\mu$ is the magnetic moment and $h$ is the field strength \cite{gentile1998large}. Specifically, the stationary distribution corresponding to $\sigma = (-1,1)^*$ is $p = \frac{1}{2}(1-b,1+b)$. Meanwhile $\lVert p \rVert^2 = (1+b^2)/2$ and $\gamma_1 = \beta \mu h$, so $\lVert {\bm \gamma } \Vert^2 = 2(\beta \mu h)^2$. 
By \eqref{eq:twostatetinfty},
\begin{equation}
\label{eq:gitinfty}
t_\infty = \beta \left [ \frac{1 + b^2}{2} \left ( 2[\beta \mu h]^2 + 1 \right ) \right ]^{-1/2}.
\end{equation}

This turns out to be a physically reasonable characteristic timescale, as we sketch here. For $\beta << 1$, $t_\infty \approx \sqrt{2} \beta$; for $\beta >> 1$, $t_\infty \approx 1/\sqrt{2} \mu h = \sqrt{2}/\Delta E$, where $\Delta E$ is the energy gap between the spin states. In both regimes $t_\infty$ is asymptotically proportional to the inverse of the natural energy scale, \emph{and in fact the constants of proportionality are the same in both regimes}. 

\begin{svgraybox}
Because mixing times are typically of the same order as inverse energy gaps, such a choice for $t_\infty$ is consistent with our overall arguments and physically justified. 
\end{svgraybox}

A routine calculation 
shows that the $L^2$ mixing time of the spin is $(2a)^{-1}$. With this in mind, an \emph{Ansatz} such as $t_\infty = (2a)^{-1}$ removes any remaining freedom in the $H$-picture and provides a plausible basis for recapturing (most of) the physical context of the spin from its statistical behavior. 
\footnote{
The natural recurrence time $4a^{-1}(1-b^2)^{-1}$ was previously considered in \cite{ford2006surfaces} as a candidate for $t_\infty$ for a single Glauber-Ising spin: however, such a choice turns out to be physically inappropriate, not least due to inconsistency with constraints imposed by intensivity.
} 
In particular, it requires a specific relationship between the spin flip rate $a$ (the physical import of which has usually been ignored) and the physical parameters $\beta$ and $\Delta E$. While we are unaware of any results that might inform the validity of this specific relationship--equivalently, the just-mentioned \emph{Ansatz}--in a single-spin system, considerations along present lines suggest an experimental framework for evaluating it. 

It would be of interest to determine to what extent timescales obtained along the lines of the present section might yield similar results for more general systems. However beyond this simple example such a task becomes difficult: even in the equilibrium case the analysis of timescales is nontrivial.

\subsubsection{\label{sec:MarkovProcesses}Markov processes}

An obvious application is to well-behaved but not necessarily reversible Markov processes specified by a transition (discrete time) or generator (continuous time) matrix on a finite state space. The invariant distribution $p$ is given as a left eigenvector of the relevant matrix. In the present context and perhaps more generally, a plausible candidate for $t_\infty$ is furnished by a modified $L^2$ mixing time: see \S \ref{sec:L2markov}.

For example, examination of Anosov systems, a single Glauber-Ising spin, and product systems (see sections \S \ref{sec:Anosov}, \S \ref{sec:TwoStateSystems}, and \S \ref{sec:ProductSystems}, respectively) all suggest a choice for $t_\infty$ along the lines of a mixing or similar timescale on physical grounds. We note that the first two of these examples have an essentially Markovian character, and the third is examined in the same spirit.

\subsubsection{\label{sec:Synchronization}Synchronization}

It is well known that many collections of mutually coupled subsystems synchronize in various senses for sufficiently large coupling. For a review of the most interesting case of chaotic systems, see \cite{boccaletti2002synchronization}. 

An interesting application of the intensivity of the effective temperature in this regard where the subsystems are taken to be identical except for their natural frequencies but also mutually interacting is a thermodynamically motivated \emph{Ansatz} for synchronization frequencies. Essentially, it is natural to view the specific process of chaotic synchronization as a particular case of the more general implied process of effective thermal equilibration. 

Without loss of generality, let the natural frequencies of subsystems be given by $\omega_m := c_m \omega_0$. Suppose furthermore that the $m$th subsystem has an effective temperature of $\beta_m^{-1}$ when uncoupled (note that although we have not identified a probability measure on the subsystem's phase space, our present considerations do not really depend on this). A trivial intensivity argument (cf. \S \ref{sec:BasicProductResults}) suggests that the synchronized/equilibrated system should then have the effective temperature $\beta_*^{-1} = \langle \beta \rangle_h^{-1}$, where $\langle \cdot \rangle_h$ denotes a harmonic mean. From the general scaling of $\beta$ with $t_\infty$, we get that $\beta_m$ scales as $c_m^{-1}$, and in turn that $\beta_*$ varies as $\langle c \rangle_a$, where $\langle \cdot \rangle_a$ denotes an arithmetic mean. This leads finally to the \emph{Ansatz} that the synchronization frequency should also vary as $\langle c \rangle_a$.

As a nontrivial example where this \emph{Ansatz} is validated, consider a system of Kuramoto oscillators \cite{acebron2005kuramoto} determined by the dynamical equations
\begin{equation}
\label{eq:kuramotoeqn}
\dot \theta_m = c_m \omega_0 + \sum_{m'} K_{mm'} \sin(\theta_{m'} - \theta_m)
\end{equation}
where $K$ is symmetric. This is a special case of the model considered in Theorem V.1 of \cite{dorfler2012synchronization}, which gives that (under some restrictions) the individual instantaneous oscillator frequencies synchronize to 
\begin{equation}
\label{eq:kuramotosynfreq}
\dot \theta_* = \omega_0 \langle c \rangle_a.
\end{equation}
That is, the scaling behavior of $\beta$ yields an \emph{Ansatz} that anticipates the synchronization result \eqref{eq:kuramotosynfreq}. 

We note finally that Theorem V.1 of \cite{dorfler2012synchronization} may suggest how to assign weights to inhomogeneous systems in a way appropriate to the overall framework of the present discussion.

\subsection{\label{sec:RemarksStatPhys}Remarks}

As we have seen, the form of equations \eqref{eq:explicitFH1}-\eqref{eq:explicitFt2} are dictated by very general physical considerations. No appeals to (e.g.) detailed balance or maximum entropy are necessary, and most of the derivation is essentially mathematical. 

In the setting of Anosov systems (see \S \ref{sec:Anosov}) the effective temperature has a purely dynamical basis rooted in the SRB measure. This dynamical grounding of the effective temperature is an important indication of its physical relevance \cite{cohen2002statistics, cohen2008entropy}. However, it can be still applied without reference to dynamics. For example, if the system under consideration is not stationary but $p$ and $t_{\infty}$ vary with time sufficiently slowly as to remain well-defined, then so will $\beta$ and $E$, and the language of equilibrium statistical physics will still be adequate. That is, there is no need for (e.g.) detailed balance or a maximum-entropy variational principle to be satisfied in order for $\beta$ to be well-defined: equation \eqref{eq:temperature} can be taken as an extension of the language of equilibrium statistical physics. Though the details of how $p$ and $t_{\infty}$ should be calculated or estimated are important and nontrivial, such questions of data analysis are properly distinct from our present considerations. 

While the continuity of $\beta$ w/r/t $t$ suffices to indicate the relevance of the present construction for quasi- and near-equilibrium systems, its scope is considerably more general. That said, the application of \eqref{eq:explicitFH1}-\eqref{eq:explicitFt2} to most physically interesting systems is highly nontrivial. For example, the nonstationarity of nonequilibrium spin systems introduces significant difficulties, while the equilibrium case is of limited interest beyond illustrative purposes. 

In practice, obtaining the appropriate $t_\infty$ presents a challenge (with a concomitant reward) that is not generally encountered in other approaches for the statistical characterization of physical and/or complex systems. In the equilibrium setting, this timescale dependence may be inverted to enforce consistency with traditional statistical physics while preserving a universal choice of scale for $\beta$. 

\begin{svgraybox}
That said, we may presently entertain the attractive possibility that a universal recipe for $t_\infty$ may exist, in terms of (e.g.) an ideal gas coupling and/or a modified $L^2$ mixing time. 
\end{svgraybox}

Apart from the distinguishing features introduced by involving the timescale $t_\infty$, at this point it should be clear that the effective temperature bears loose analogies to Shannon entropy both in its functional dependence and its physical content. Though an information-theoretic interpretation of the effective temperature is not obvious, its relevance to data analysis has been demonstrated elsewhere in the context of computer network traffic analysis; meanwhile, an examination of thermodynamical analogies in the information theory of discrete memoryless channels is also presently being undertaken and holds promise for illuminating the nature and role of $t_\infty$.

\subsubsection{\label{sec:ObstructionToAnalogues}Obstruction to analogues for (e.g.) Bose-Einstein or nonextensive statistics}

Consider a notional alternative to the Gibbs distribution of the form 
\begin{equation}
p_k \equiv f(-\gamma_k)/\zeta.
\end{equation}
Now $-\gamma_k = f^{-1}(\zeta p_k)$ and if $\sum_j f^{-1}(\zeta p_j) = 0$, then 
\begin{equation}
\gamma_k = \frac{1}{n} \sum_j \left [ f^{-1}(\zeta p_j) - f^{-1}(\zeta p_k) \right ].
\end{equation}

The derivation of the formula for $\beta$ depends in an essential way on the existence of a relation of the form $f^{-1}(\zeta p_j) - f^{-1}(\zeta p_k) \equiv g(p_j,p_k)$. If such a relation holds, differentiating both sides w/r/t $\zeta$ gives that $y \cdot D_y(f^{-1}(y))$ is constant. It follows that $f(x) = \exp(cx + c')$ for some $c, c'$: this amounts to reproducing the Gibbs distribution. In other words, generalizations of the effective temperature building on e.g. Bose-Einstein or nonextensive statistics cannot be constructed along obvious lines.


\subsubsection{\label{sec:ContinuousDistributions}Naive requirements for continuous distributions}

Dealing with a more general reference probability measure $\nu$ than a normalized counting measure is straightforward provided that the physical measure $\mu$ is absolutely continuous w/r/t $\nu$ and that both $p \equiv d\mu/d\nu$ and $\log p$ are in $L^1(\nu) \cap L^2(\nu)$. 

To see this, recall that $f \in L^q(\nu)$ iff $\lVert f \rVert_q := (\int |f|^q \ d\nu)^{1/q} < \infty$. Using the additional shorthand $\ell := \log p$, we have that the analogue to \eqref{eq:E0} is $\int E \ d \nu = 0$, from which it follows that $-\int \ell \ d \nu = \log Z$. Further brief manipulations yield the generalization of \eqref{eq:betaEk}, namely $\gamma(x) = \int \ell \ d\nu - \ell(x)$, and we also have that $\lVert \gamma \rVert^2_2 = \lVert \ell \rVert^2_2 - (\int \ell \ d\nu)^2$. This is well-defined if $\ell \in L^1(\nu) \cap L^2(\nu)$. If moreover we have that $p \in L^2(\nu)$ then the natural analogue of \eqref{eq:explicitFH1} is well-defined. Note that $p \in L^1(\nu)$ since $\int p \ d\nu = \int d\mu \equiv 1$.

However, these integrability conditions are rarely met in situations of interest. Even more fundamentally, SRB measures are typically not absolutely continuous w/r/t underlying Riemannian measures. For this reason the application to Anosov-like systems in \S \ref{sec:Anosov} is necessarily more involved.

\subsubsection{\label{sec:OverallScale}The choice of overall scale and the zeroth law}

The requirement that $\beta$ equal the actual physical inverse temperature for equilibrium systems strongly constrains $t_\infty$, and (modulo issues of state space discretization) completely specifies the product $S^{-1} t_\infty$ appearing in, e.g., \eqref{eq:scaling}. That is, mandating equivalence of the effective and actual temperature wherever possible links $S$ and $t_\infty$. It is clear that we may choose either $S$ or $t_\infty$ to be system-independent at the cost of admitting at least the possibility for system-dependence on the other. However, we have (without loss of generality) enforced the overall choice of scale $S \equiv 1$.
\footnote{
While $S \equiv S(n) \not \equiv 1$ might appear to be a reasonable middle ground, e.g., $S = \sqrt{n}$, considerations of the sort described elsewhere for $t_\infty$ also militate against this.
}

Subject to this choice of overall scale, the ultimate physical significance of $\beta$ will necessarily depend on the (as yet unknown) degree to which we can have $\beta$ equal the actual physical inverse temperature in different equilibrium systems without requiring $t_\infty$ to have some system-dependent definition (or to absorb some system-dependent constant) to compensate. 

Nevertheless, even in the most pessimistic case of a completely system-dependent overall scale, the effective temperature (or a ratio of effective temperatures with the same choice of overall scale) could still be fruitfully used to ``internally'' characterize individual systems that vary in time sufficiently slowly for $p$ and $t_\infty$ to remain well-defined, or to compare multiple systems that are identical save for some parametric dependence over a statistical ensemble (and perhaps especially an ensemble which permits perturbations from equilibrium). In fact, the former situation obtains in the analysis of, e.g., experimental data with long-timescale variability. 

Therefore, a system-independent choice of scale is not necessary to establish that there is some physical relevance for $\beta$, but only the scope of that relevance. However, we point out that at least a weak degree of system-independence is exhibited for the examples in the preceding paragraph as well as the ideal gas on a surface of constant negative curvature, as the genus does not appear to affect either $t_\infty$ or the value of $\beta$ (see Figures \ref{fig:geodesicbeta} and \ref{fig:geodesicbeta2}).

In general, there appear to be two basic avenues to addressing concerns of system-dependence of scale (say, as manifested in $t_\infty$ with $S \equiv 1$), which even from a pessimistic point of view would turn out to be at least approximately equivalent in some circumstances for the reasons cited just above. 

The first avenue is, in the absence of any other generically useful and identifiable recipe for computing $t_\infty$ \emph{a priori}, to take the requirement that $\beta$ equal the actual physical inverse temperature in equilibrium to operationally \emph{define} $t_\infty$. The example of a single Glauber-Ising spin in \S \ref{sec:TwoStateSystems} indicates how a $t_\infty$ obtained in this way can be physically meaningful. Taking this avenue might help to place physical constraints on and even select a preferred system-independent characterization of $t_\infty$ (e.g., as a modified $L^2$ mixing time) valid both in and beyond equilibrium. 

The second avenue is more difficult and ambitious, but likely also more sound. It involves coupling systems to an ideal gas and enforcing the constraint $\beta^{(\text{sys})} = \beta^{(\text{sys}+\text{gas})} = \beta^{(\text{gas})} \equiv \beta$ in a suitable coupling and/or large $N$ limit for the gas. That is, this approach takes the zeroth law as an \emph{Ansatz}. It can be hoped that it might be possible in principle to infer a well-defined $t_\infty^{(\text{sys})}$ from an implied timescale $t_\infty^{(\text{sys}+\text{gas})}$,
\footnote{
For considerations affecting systems with multiple independent characteristic timescales, see \S XVIII of \cite{huntsman2010anosov}.
} 
subject to the temperature constraint above. Better still would be a comparatively simple recipe for determining this $t_\infty^{(\text{sys})}$ such as those proposed in \S \ref{sec:L2markov}. 

\subsubsection{\label{sec:Coda} Coda}

Though the typical state of affairs is for the ordinary temperature to be regarded as an environmental parameter in calculations, the logic may be largely turned on its head: in many cases we can directly obtain an effective temperature and (re)construct an effective Hamiltonian from the behavior of a system. In this way the idiom of equilibrium statistical physics may be extended for many applications in nonequilibrium steady states and problems in data analysis. Finally, while a philosophical study of thermometry notes that ``there are complicated philosophical disputes about just what kind of quantity temperature is" \cite{chang2004inventing}, we hope to catalyze investigations in this direction.

\section{\label{sec:Anosov} Application to Anosov systems}

The examples and applications in \S \ref{sec:ExamplesAndApplications} of the framework of \S \ref{sec:2} are only a partial list. More substantive efforts have been or are focused on, e.g. characterization of computer network traffic \cite{huntsman2009effective}, physical correspondences in the information theory of discrete memoryless channels, and data science. Here, however, we will discuss an application to mixing Anosov systems in some detail, as this realistic physical context underlines the equivalence of statistical ($t$) and physical ($H$) descriptions. 

A more comprehensive treatment of the material in this section is in \cite{huntsman2010anosov}.

\subsection{\label{sec:AnosovOverview} Overview}

Essentially, a \emph{mixing Anosov system} is a well-behaved uniformly hyperbolic dynamical system (see Figure \ref{fig:AnosovFlow} for a schematic and \cite{bowen1975equilibrium,katok1997introduction,chernov2002invariant,jiang2004mathematical} for background elaborating on \S \ref{sec:AnosovBackground}). Such systems are particularly relevant to statistical physics: indeed, the so-called \emph{chaotic hypothesis} is that many-particle systems are essentially mixing and Anosov insofar as their macroscopic properties are concerned \cite{gallavotti1995stationary,gallavotti1999statistical}. Underpinning this conjecture is the existence (for a compact phase space, which we assume for convenience) of the SRB measure $\mu_{SRB}$, an invariant physical probability measure generalizing the microcanonical ensemble \cite{young2002srb}.

The two quintessential examples of Anosov systems (both mixing) are the discrete-time Arnol'd-Avez cat map (more generally, a hyperbolic toral automorphism) and the geodesic flow describing a free particle on a compact surface of constant negative curvature. 
\footnote{
A discrete-time version of the latter is obtained by considering a Poincar\'e or timing map.
}

We first outline how to deal with the continuous phase space of a mixing Anosov system in a natural way. One of two key observations in this regard is that the hyperbolic dynamics furnish a physically natural family of phase space discretizations, called \emph{Markov partitions}. 
We recall that a Markov partition $\mathcal{R}$ for an Anosov diffeomorphism $T$ is a decomposition of phase space into so-called \emph{rectangles} $R_j$ with local product structure compatible with the hyperbolic structure of $T$ and such that the images of rectangles under $T$ stretch completely across the original rectangles in the unstable direction and vice versa for $T^{-1}$ (see Figure \ref{fig:MarkovSchematic}).
Here the probability distribution corresponding to a Markov partition $\mathcal{R}$ is simply given by $p_j := \mu_{SRB}(R_j)$ for each $R_j \in \mathcal{R}$, and the $L^2$ mixing time of the system is taken as (a placeholder/approximation for) $t_\infty$. 
\footnote{
The \emph{Ruelle-Bowen hypothesis} that mixing Anosov flows are exponentially mixing (with respect to H\"older observables and equilibrium measures with H\"older potentials) \cite{dolgopyat1998decay,liverani2004contact,butterley2020open,tsujii2020smooth} would further support the existence of a mathematically and physically natural $t_\infty$.
}
Ergodicity ensures that this is equivalent to considering the time series of indices of rectangles that contain a test particle.

The second key observation for dealing with the continuous phase space is that its geometrical measure (i.e., the Riemannian and not the SRB measure) determines a procedure for obtaining \emph{greedy refinements} of any initial Markov partition (see Figures \ref{fig:MarkovSchematic} and \ref{fig:Greedy0}-\ref{fig:Greedy2}). 
\footnote{
While not every Anosov system will preserve a natural Riemannian measure, the archetypes we consider do: for hyperbolic toral automorphisms, this is just the (pushforward of) Lebesgue measure, and for the geodesic flow on a surface of constant negative curvature it is the Liouville measure. More generally, so-called \emph{conservative} diffeomorphisms preserve a natural Riemannian measure (and diffeomorphisms in general preserve an equivalence class of measures) \cite{Wilkinson}. A wide class of conservative diffeomorphisms is furnished by Hamiltonian systems, and it is natural to couch the otherwise implicit notion of a ``natural'' Riemannian measure in this context.
}

These greedy refinements are Markov by construction, and a subsequence of them is maximally uniform w/r/t the geometrical measure. Therefore, from the physical point of view, intermittent greedy refinements of an initial Markov partition are particularly natural. Indeed, their uniformity enables the proof of a finite value for $\lim \inf \beta$ for the cat map, where the limit inferior is taken over successive greedy refinements.
\footnote{
Based on some explicit calculations for the cat map, it may be that $\lim \inf \beta$ is independent of the choice of initial Markov partition: however, if this turns out not to be the case, a suitable extremum over Markov partitions may be considered.
}
Furthermore, considering greedy refinements for the geodesic flow also gives compelling numerical evidence that $\lim \inf \beta$ is not only finite but independent of global structure (see Figures \ref{fig:geodesicbeta} and \ref{fig:geodesicbeta2}). 

\begin{figure}[htbp]
\includegraphics[trim = 5mm 0mm 10mm 5mm, clip, width=\textwidth,keepaspectratio]{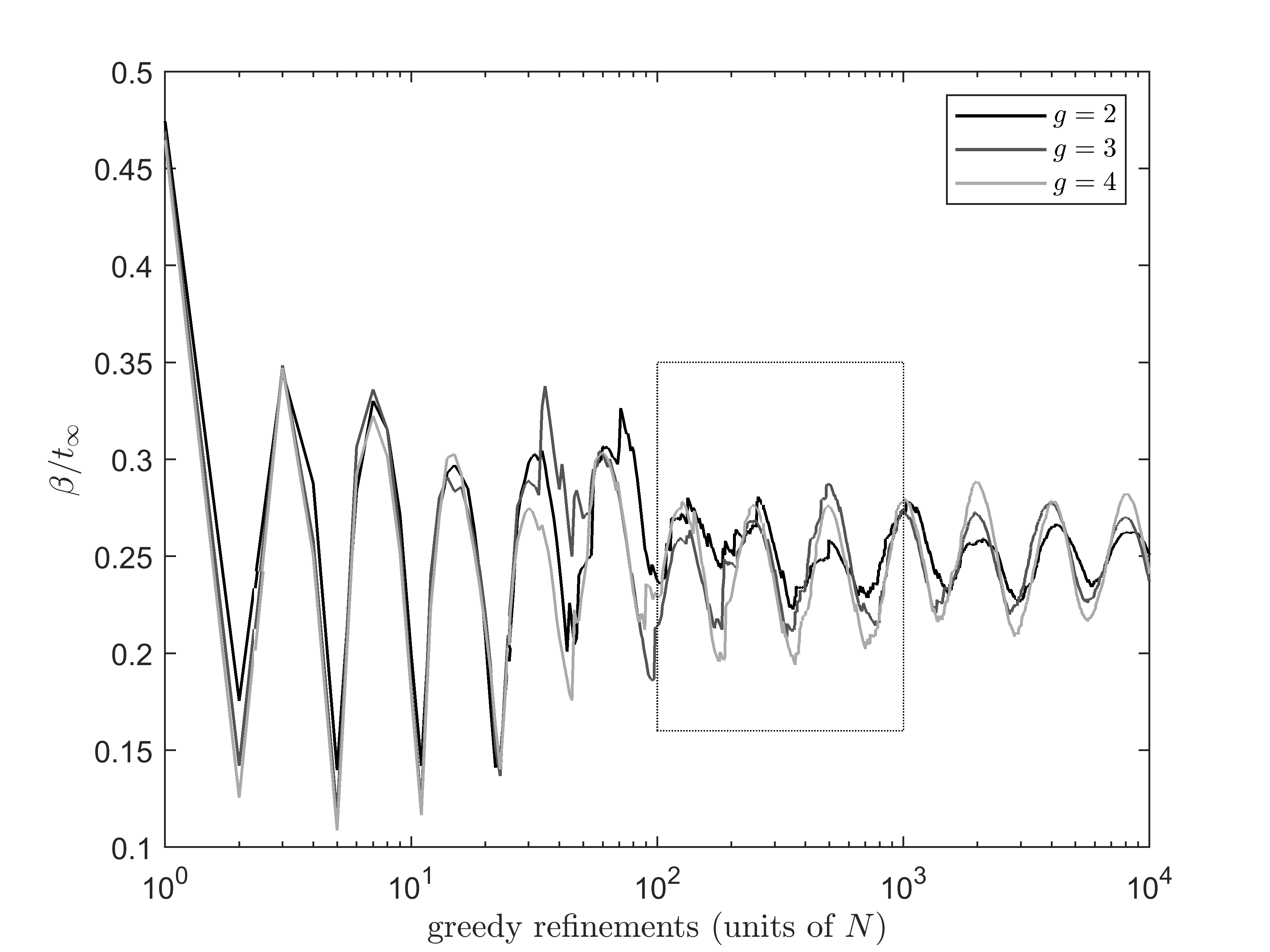}
\caption{ \label{fig:geodesicbeta} $\beta/t_\infty$ under successive greedy refinements of an initial Markov partition for a map isometrically topologically conjugate to a timing or Poincar\'e map of the geodesic flow on a compact surface of constant negative curvature with genus $g=2,3,4$. Note the logarithmic horizontal scale: here $N=8g-4$ is related to the number of rectangles in the initial partition. The inset box indicates axis limits for Figure \ref{fig:geodesicbeta2}.}
\end{figure}

\begin{figure}[htbp]
\includegraphics[trim = 5mm 0mm 10mm 5mm, clip, width=\textwidth,keepaspectratio]{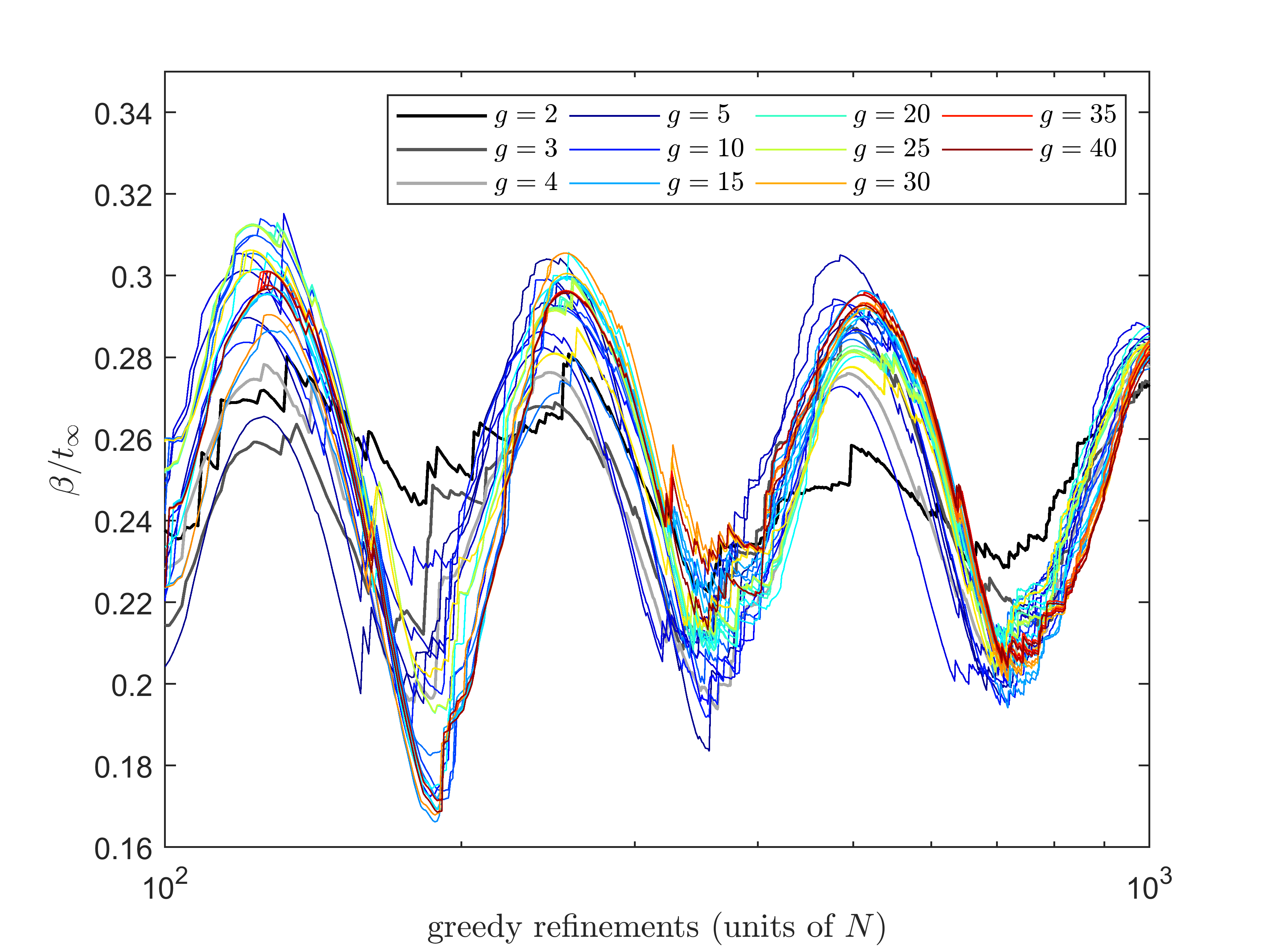}\caption{ \label{fig:geodesicbeta2} As in Figure \ref{fig:geodesicbeta} for $g=2,\dots,40$. Axis limits correspond to the inset box in Figure \ref{fig:geodesicbeta}. The existence of a nontrivial limit inferior independent of $g$ is strongly suggested by such numerical results. Meanwhile, both the $L^2$ mixing time for sufficiently well-behaved observables and the inverse topological entropy turn out to be plausible genus-independent approximations for $t_\infty$.}
\end{figure}

%
%
%

We now turn to briefly describing how the preceding results bear on generic systems of interest in statistical physics. A product system formed from copies of the geodesic flow with canonically distributed initial conditions is just an ideal gas.
\footnote{
The consideration of a product system formed from statistically identical subsystems places a very strong intensivity constraint on the form of $t_\infty$ that forces this quantity to be similar to, yet necessarily distinct from, the $L^2$ mixing time. 
}
By considering this gas as an environment and weakly coupling a subsystem to it, our focus shifts from a generalized microcanonical ensemble to the canonical ensemble, and a definition of a generic subsystem's effective temperature in terms of that of its environment. 

Such a procedure should often if not always give physically reasonable and self-consistent results, as we proceed to sketch. The results of \cite{amaricci2007analyticity} and the general phenomenon of structural stability of Markov partitions for Anosov systems show that the effect on $\beta/t_\infty$ of coupling a Gaussian thermostat to the geodesic flow is analytic in the strength of the coupling. 
\footnote{
In the thermodynamical limit, we expect dynamics to be insensitive to the details of thermostatting, i.e., the various SRB measures should tend to the same limit \cite{gallavotti1999statistical,ES,bonetto2006chaotic,gallavotti2009thermostats,gallavotti2010nonequilibrium,gallavotti2010thermodynamic}.
}
It is also reasonable to expect that $t_\infty$ will exhibit a similar regularity as a function of the coupling based on (e.g.) the stability of rapid mixing \cite{field2007stability}. Many other examples are known in which well-behaved Anosov systems exhibit considerable stability of SRB measures 
and mixing times w/r/t perturbations, and it is reasonable to expect such behavior in general for physically relevant cases.

\begin{svgraybox}
Taken together, these observations strongly suggest the existence and essential uniqueness of a physically preferred effective temperature and concomitant energy function intrinsic to generic mixing Anosov systems. Moreover, they suggest an avenue for extending Ruelle's thermodynamical formalism \cite{Ruelle1} into a more comprehensive theory of statistical physics for nonequilibrium steady states obeying the chaotic hypothesis. 
\end{svgraybox}

\subsection{\label{sec:AnosovBackground} Background on Anosov systems}

A smooth endomorphism $T$ of a Riemannian manifold (``phase space'') is an \emph{Anosov map} if it is both
\begin{itemize} 
	\item \emph{uniformly hyperbolic}, i.e. at every point $x$ there are transverse local stable and unstable surfaces on which points respectively converge and diverge exponentially at a rate independent of $x$; and
	\item \emph{invariant}, i.e. the tangent spaces to these surfaces are mapped by the derivative of $T$ into the tangent spaces to the corresponding surfaces at $Tx \equiv T(x)$.
\end{itemize}
If the global stable and unstable surfaces of $T$ are dense, then $T$ is also said to be \emph{mixing}. An \emph{Anosov flow} is a continuous-time analogue of an Anosov map with a neutral surface transverse to the time evolution as schematically indicated in Figure \ref{fig:AnosovFlow}. We refer to both Anosov maps and flows as \emph{Anosov systems}.

\begin{figure}[htbp]
\centering
	\includegraphics[trim=0mm 5mm 0mm 0mm, clip, width=5cm,keepaspectratio]{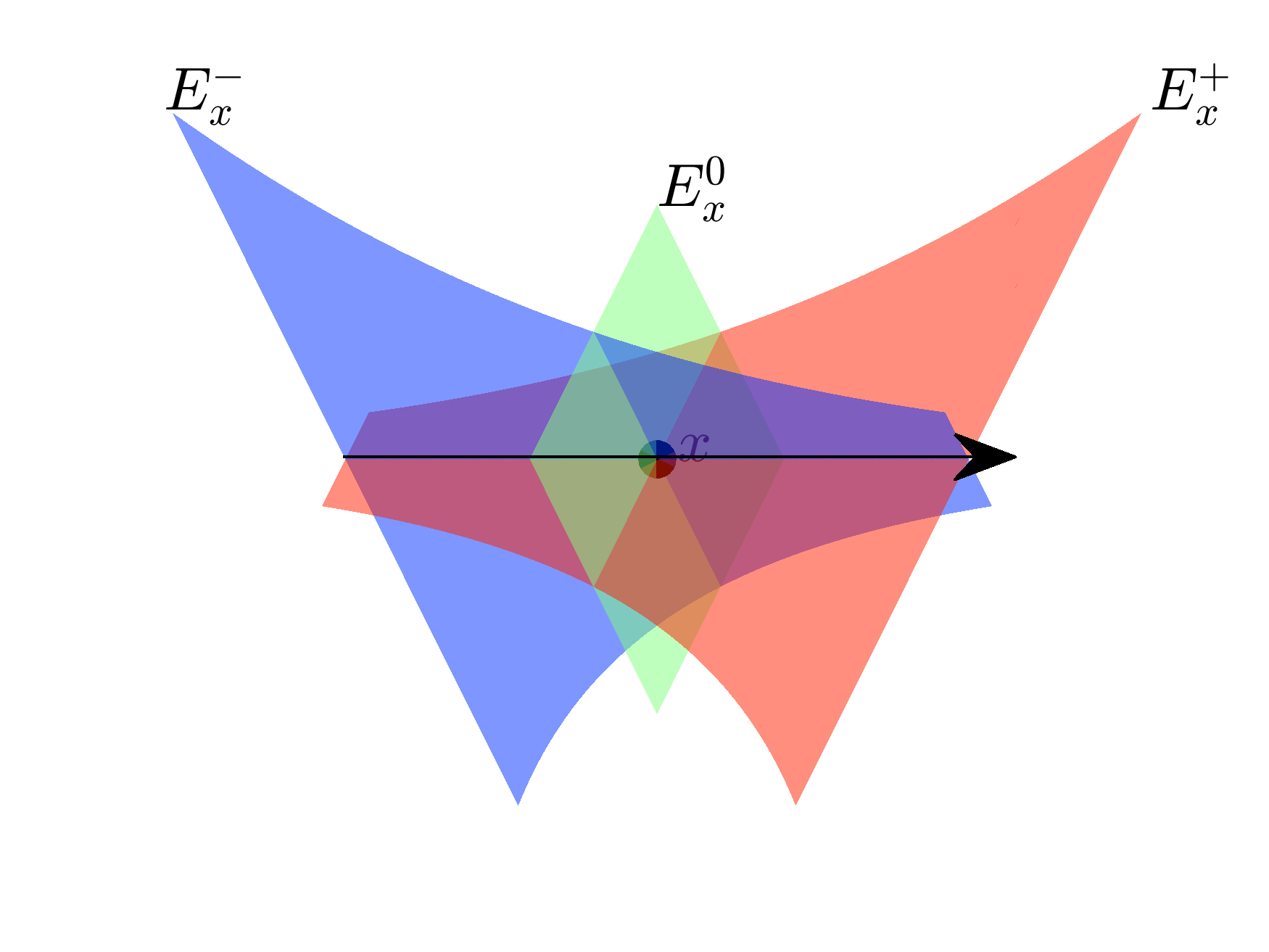}
\caption{ \label{fig:AnosovFlow} Schematic of an Anosov flow, with respective local stable, unstable, and neutral surfaces $E_x^-$, $E_x^+$, and $E_x^0$.}
\end{figure}

Anosov systems enjoy natural discretizations of their phase spaces. These discretizations, called \emph{Markov partitions}, are particular configurations of \emph{rectangles}. 
\footnote{
Strictly speaking, Anosov flows require a related notion called a \emph{Markov section}, but this distinction can be mostly ignored here. See \cite{chernov2002invariant} for details.
}
A rectangle $R$ is a subset of phase space such that the intersection of a local stable surface and a local unstable surface consists of a single point also in $R$: i.e., there is a local product structure compatible with $T$. 
\footnote{
Rectangles in the present context are not, and should not be generically thought of as, right-angled quadrilaterals (indeed, rectangles are generically fractal in character). However, the specific examples we consider will be right-angled quadrilaterals.
}
A partition $\mathcal{R} = \{R_j\}_{j=1}^n$ of phase space into rectangles is Markov if (whenever these sets intersect) the images $TR_j$ stretch completely across $R_k$ in the unstable direction and $R_k$ stretches completely across $TR_j$ in the stable direction, as schematically indicated in Figure \ref{fig:MarkovSchematic}. The utility of the coarse-graining of phase space that Markov partitions provide is largely attributable to the fact that (via the theory of symbolic dynamics) they allow Anosov systems to be treated in much the same way as a spin system \cite{beck1995thermodynamics,gallavotti1999statistical}. This also highlights the relevance of both Anosov and spin systems as extremely generic (or as the chaotic hypothesis argues, completely generic) models for statistical physics.
\footnote{
For example, a unique SRB measure corresponds to the absence of phase transitions in one-dimensional short-ranged spin models \cite{gallavotti1999statistical}. This Anosov-spin system correspondence also suggests the construction of $d$-dimensional lattices of coupled maps corresponding to $(d+1)$-dimensional spin systems capable of exhibiting phase transitions \cite{chazottes2005dynamics}. 
}

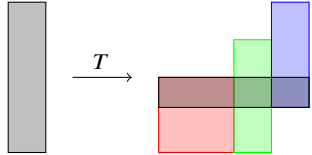
\begin{figure}[htbp]
\centering
	\begin{tikzpicture}[scale = 1,every node/.style={transform shape}]
	\begin{scope}[fill opacity=0.25]
		\draw[draw=black,fill=black] (-2,0) rectangle (-1.5,2);
		\draw[draw=red,fill=red] (0,0) rectangle (1,1);
		\draw[draw=green,fill=green] (1,0) rectangle (1.5,1.5);
		\draw[draw=blue,fill=blue] (1.5,.6) rectangle (2,2);
		\draw[draw=black,fill=black] (0,.6) rectangle (2,1);
	\end{scope}
	\node (v0) at (-1.25,1) {};
	\node (v1) at (-.25,1) {};
	\path [->] (v0) edge node [above] {$T$} (v1);
	\end{tikzpicture}
\caption{ \label{fig:MarkovSchematic} Schematic of the action of an Anosov map on a rectangle that is part of a Markov partition. The image of the rectangle stretches precisely across other rectangles in the partition. If we suppose that the phase space measure corresponds to the Lebesgue measure of the rectangles, then a \emph{greedy refinement} of the gray rectangle can be obtained by drawing a line at the inverse image of the intersection of the gray, red, and green rectangles. That is, a greedy refinement of a Markov partition $\mathcal{R}$ for an Anosov map $T$ starts by considering the forward image under $T$ of a rectangle $R_j \in \mathcal{R}$ with maximal geometrical/phase space (vs. physical/SRB) measure. The intersection of the boundary of $\mathcal{R}$ and $TR_j$ determines subrectangles of $T\mathcal{R}$ that in turn determine various refinements of $\mathcal{R}$ under $T^{-1}$. A greedy refinement has maximal entropy w/r/t the geometrical measure.
}
\end{figure}

\subsection{\label{sec:CatMap} The cat map}

The simplest example of an Anosov map is the Arnol'd-Avez \emph{cat map} defined on the unit torus via $T_A x := Ax \mod 1$, where $A = \left(\begin{smallmatrix}
		2 & 1 \\
		1 & 1
		\end{smallmatrix} \right)$.
\footnote{
The cat map corresponds to unit-frequency projections for the Hamiltonian $\mathcal{H}_A(X,P) = K(P^2-X^2 + XP)$ with $K = \sinh^{-1}(\sqrt{5}/2)/\sqrt{5}$.
}
The eigendecompositon of $A$ determines the stable and unstable directions. The eigenvalues are $\lambda_\pm = \phi^{\pm 2}$, where $\phi = \frac{1+\sqrt{5}}{2}$ is the golden ratio. The corresponding eigenvectors are $e_- = (s,-c)^*$ and $e_+ = (c,s)^*$, where $c = 1/\sqrt{3-\phi}$ and $s = \sqrt{1-c^2}$. Because these eigenvectors have irrational slopes, the stable and unstable surfaces are dense on the torus, so the cat map is mixing.

More generally, matrices in $GL(n,\mathbb{Z})$ with no eigenvalues in $S^1$ correspond to \emph{hyperbolic toral automorphisms} which are also Anosov maps. In dimension $n = 2$, rectangles for these maps are geometrically unions of parallelograms.

\subsection{\label{sec:MarkovGreedy} Markov partitions and greedy refinements}


There are many Markov partitions for the cat map: Figure \ref{fig:MarkovPartitions} shows three, respectively denoted $\mathcal{R}_A$, $\mathcal{R}'_A$, and $\mathcal{R}''_A$.
\begin{figure}[htbp]
\centering
	\includegraphics[trim = 84mm 111mm 84mm 111mm, clip, width=.3\textwidth,keepaspectratio]{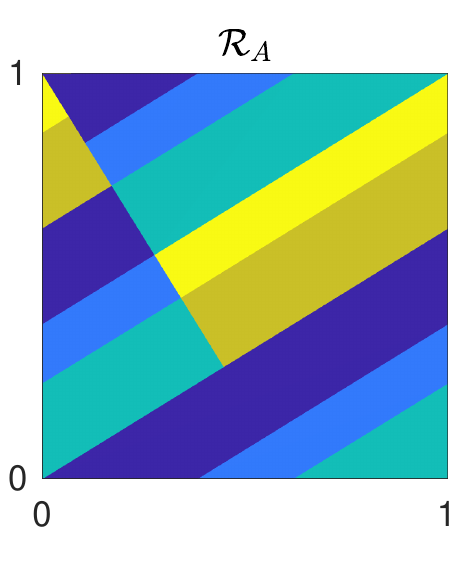} 
	\includegraphics[trim = 84mm 111mm 84mm 111mm, clip, width=.3\textwidth,keepaspectratio]{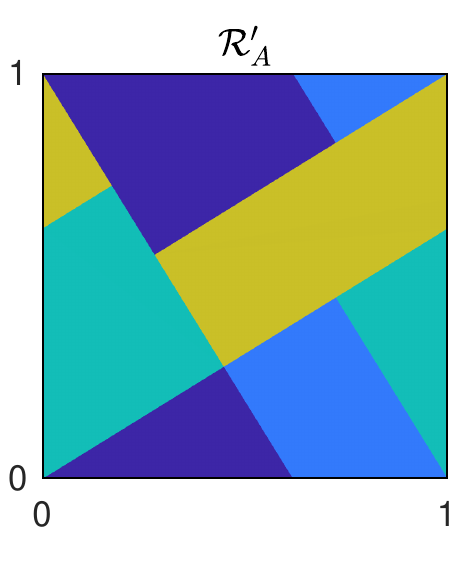} 
	\includegraphics[trim = 84mm 111mm 84mm 111mm, clip, width=.3\textwidth,keepaspectratio]{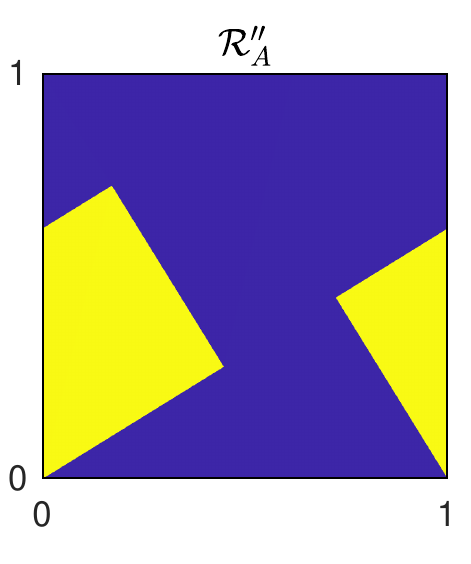}
\caption{ \label{fig:MarkovPartitions} Three Markov partitions for the cat map.}
\end{figure}
A Markov partition $\mathcal{R} = \{R_j\}_{j=1}^n$ induces a probability distribution $p_j := \mu(R_j)$ inherited from the physical/SRB measure $\mu \equiv \mu_{SRB}$. For the cat map (or any other hyperbolic toral automorphism), this measure is just the Lebesgue measure. 

For \emph{any} two-dimensional hyperbolic toral automorphism $T$ (including the cat map), there is a remarkable fact (for details, see \cite{huntsman2010anosov}). Let $\mathcal{R}$ be a Markov partition and let $\mathcal{R}^\vee_m$ be a refinement of $\mathcal{R}$ obtained by taking connected components of intersections of rectangles in $T^j \mathcal{R}$ for $0 \le j \le m$. Now as $m \rightarrow \infty$, $\beta/t_{\infty} =  \lVert p \rVert \cdot \sqrt{\lVert \gamma \rVert^2 +1}$ (which does not depend on $t_\infty$) converges to a finite nonzero value. 

To see why any finite nonzero limit for $\beta/t_\infty$ is nontrivial, consider the following toy example. Define $\mathcal{Y}^{(0)} = [0,1]$ and form $\mathcal{Y}^{(m+1)}$ by subdividing each interval in $\mathcal{Y}^{(m)}$ into two subintervals of relative length $q$ and $1-q$. The corresponding partitions yield $\lim \beta/t_\infty = \infty$ unless $q = 1/2$, in which case the limit is zero. Meanwhile, as we have mentioned, a naive discretization of the free particle/ideal gas has no obvious reasonable scaling limit.

The preceding results involving two-dimensional hyperbolic toral automorphisms indicate that while it might seem useful to consider $\beta = S_n \cdot t_\infty \lVert p \rVert \sqrt{\lVert \gamma \rVert^2 +1}$ with $S_n = \sqrt{n}$ so that $\beta$ is independent of $n$, this is a mirage: we should actually require $S_n$ to be constant in $n$. 

Now while detailed calculations establish that $\lim \beta/t_\infty$ depends on $\mathcal{R}$, and phase space (to say nothing of physical) measures of rectangles in $\mathcal{R}^\vee_m$ vary increasingly more as $m$ increases, there is a straightforward solution. We can construct \emph{greedy refinements} of $\mathcal{R}$ that are more physically natural by maximizing the uniformity of phase space measures at each step of the refinement process. Even when the physical measure $\mu$ and phase space measure $\nu$ disagree,
\footnote{
Typically, $\mu$ will be singular with respect to $\nu$, though for two-dimensional hyperbolic toral automorphisms both measures are equal to Lebesgue measure.
}
considering greedy refinements will tend to minimize $\beta$ and maximize entropy/minimize effective free energy. This is an indication of a generalized variational principle (in the sense of ergodic theory) that can yield a finite limit for $\beta$ even as the entropy of partitions diverges.

The construction of greedy refinements goes as follows. For a rectangle $R_j \in \mathcal{R}$ with maximal phase space measure $\nu(R_j)$, the intersection of $TR_j$ with rectangles in $\mathcal{R}$ determines subrectangles of $T\mathcal{R}$ that in turn determine various refinements of $\mathcal{R}$ under $T^{-1}$. We call such a refinement of maximal entropy (with respect to $\nu$) \emph{greedy}. In general, greedy refinements are not unique, though subsequences of them (corresponding to the result of greedily refining all of the rectangles of maximal $\nu$-measure at a time) will be. Figures \ref{fig:Greedy0}-\ref{fig:Greedy1} illustrate greedy refinements for the Markov partitions in Figure \ref{fig:MarkovPartitions}.


\begin{figure}[htbp]
\centering
	\includegraphics[trim = 30mm 130mm 25mm 125mm, clip, width=\columnwidth,keepaspectratio]{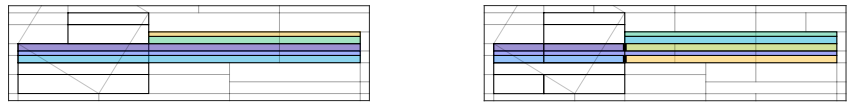}
\caption{ \label{fig:Greedy0} (L) In black, we show the Markov partition $\mathcal{R}_A$ from Figure \ref{fig:MarkovPartitions} in eigencoordinates, with both the unit square and translates in gray. The forward image $T \mathcal{R}_A$ is shown in color. There are two rectangles of maximal $\nu$-measure: the third and fifth from the top. (R) Greedily refining each of the two rectangles of maximal $\nu$-measure by taking intersections as demarcated by bold black lines.}
\end{figure}

\begin{figure}[htbp]
\centering
	\includegraphics[trim = 30mm 130mm 25mm 125mm, clip, width=\columnwidth,keepaspectratio]{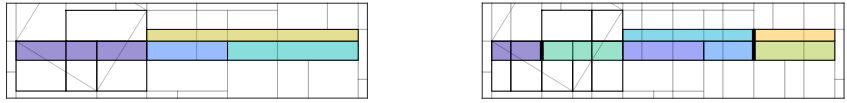}		
\caption{ \label{fig:Greedy1} (L) As in Figure \ref{fig:Greedy0}, but for $\mathcal{R}_A'$. (R) The result of two greedy refinements. Note that the result of a round of greedy refinements is not unique, though the value of the resulting $\nu$-entropy is.}
\end{figure}

\begin{figure}[htbp]
\centering
	\includegraphics[trim = 30mm 145mm 25mm 75mm, clip, width=\columnwidth,keepaspectratio]{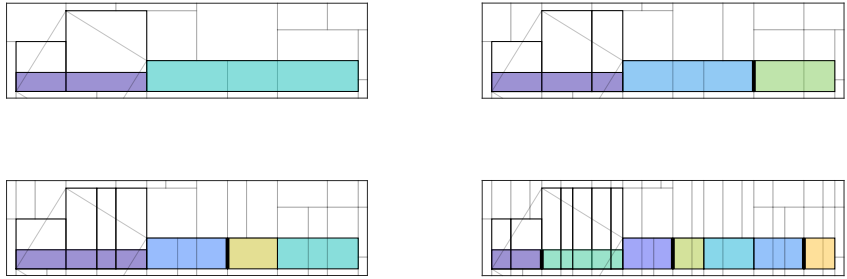}		
\caption{ \label{fig:Greedy2} Successive greedy refinements for $\mathcal{R}_A''$.}
\end{figure}

These greedy refinements stabilize the measures of rectangles. Detailed calculations show that certain greedy refinements of both $\mathcal{R}_A$ and $\mathcal{R}_A'$ contain $L_{m+1}$ and $L_{m+2}$ rectangles of relative measure $1$ and $\phi$, respectively. Here the \emph{Lucas numbers} are defined via $L_{m+2} = L_{m+1} + L_m$ with $L_1 = 1$, $L_2 = 3$. However, $\mathcal{R}_A'$ behaves differently, with maximally uniform greedy refinements containing $L_{m-1}$ and $L_m$ rectangles of relative measure $1$ and $\phi$, respectively. Nevertheless, in each of these three cases there is a common limit $\lim \beta/t_\infty \approx 0.2393$, which is probably minimal/universal for the cat map. In any event, while the detailed measures of greedy refinements depend on an initial Markov partition, we can always consider an extremum over Markov partitions with decreasing size to obtain (by construction) a unique physically natural result.

\subsection{\label{sec:FreeParticle} The geodesic flow on a surface of constant negative curvature}

The geodesic flow on a surface of constant negative curvature is the archetypal Anosov flow \cite{anosov1969geodesic,klingenberg1974riemannian}. It corresponds to the free particle Hamiltonian \cite{collet1984perturbations} $\mathcal{H} = \frac{1}{2m}\sum_{jk}g^{jk}P_j P_k$, where we use typical notation for the inverse of the metric tensor and for momenta on the cotangent bundle. For a surface of constant negative curvature, the geodesic flow is mixing, and as we shall see the effective temperature is apparently insensitive to the surface genus.

We briefly recall the details of the geodesic flow in the Poincar\'e disk model following \cite{adler1991geodesic} (see also \cite{bowen1979markov,gutzwiller1990chaos}). The differential arclength is $ds = dr/(1-r^2)$, and geodesics correspond to circular arcs intersecting $S^1$ at right angles (see Figure \ref{fig:genus2}). A surface of constant negative curvature can be obtained by identifying pairs of edges $s_j$ of a hyperbolic polygon via maps $T_j(s_j) = s_{\sigma(j)}^{-1}$. Here $s_j^{-1}$ denotes the orientation reversal of $s_j$, and the pairing $\sigma$ is defined along the lines shown for the genus 2 case of Figure \ref{fig:genus2}. 
\footnote{
Note that this pairing is not ``twisted.''
}
If there are $8g-4$ edges, this procedure yields a surface of genus $g$. Finally, the Hamiltonian is $\mathcal{H} = (1-r^2)^2 \cdot P^2/2m$.

\begin{figure}[htbp]
	\includegraphics[trim = 28mm 25mm 23mm 25mm, clip, width=.59\textwidth,keepaspectratio]{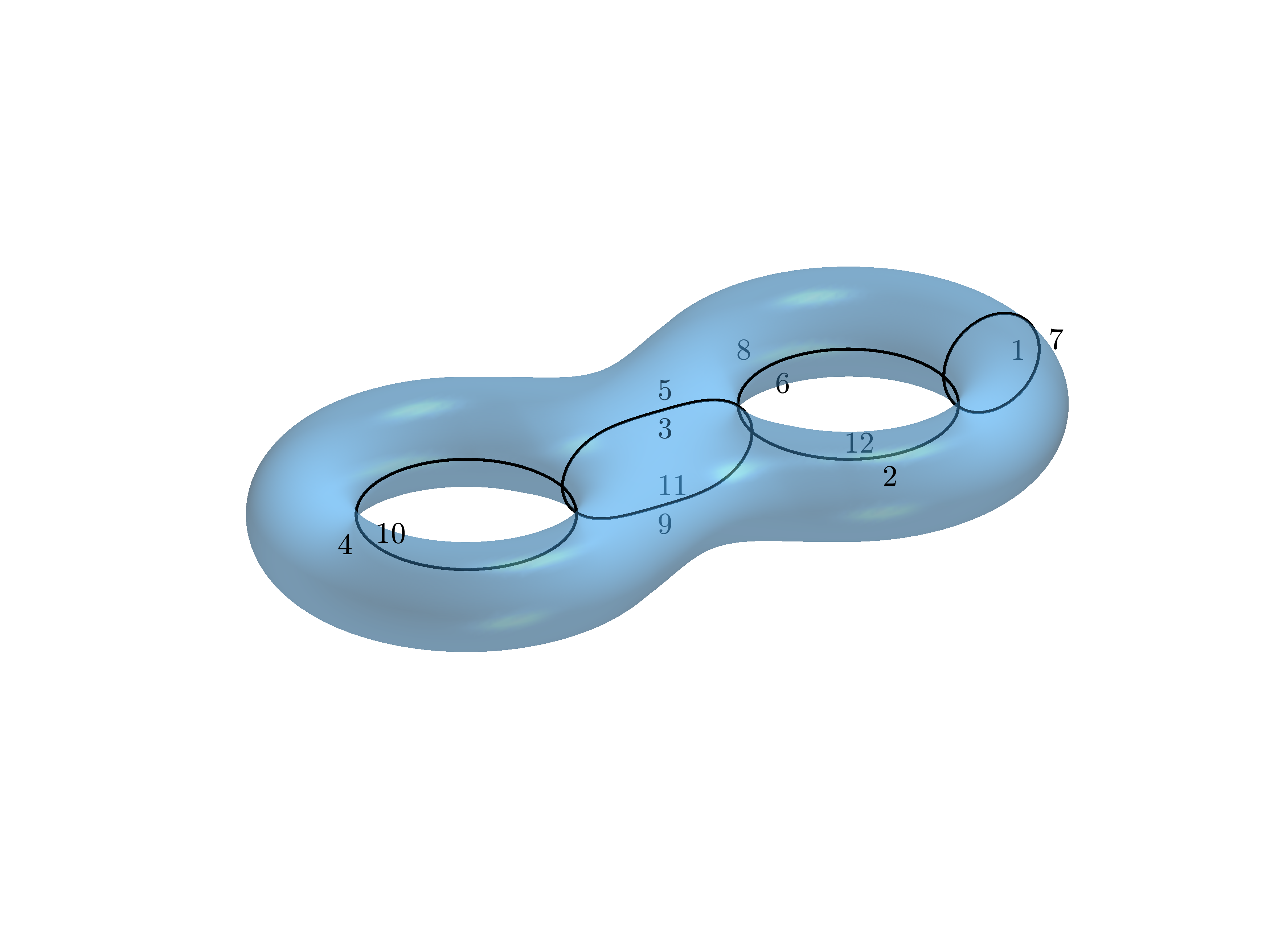}
	\includegraphics[trim = 25mm 5mm 25mm 5mm, clip, width=.4\textwidth,keepaspectratio]{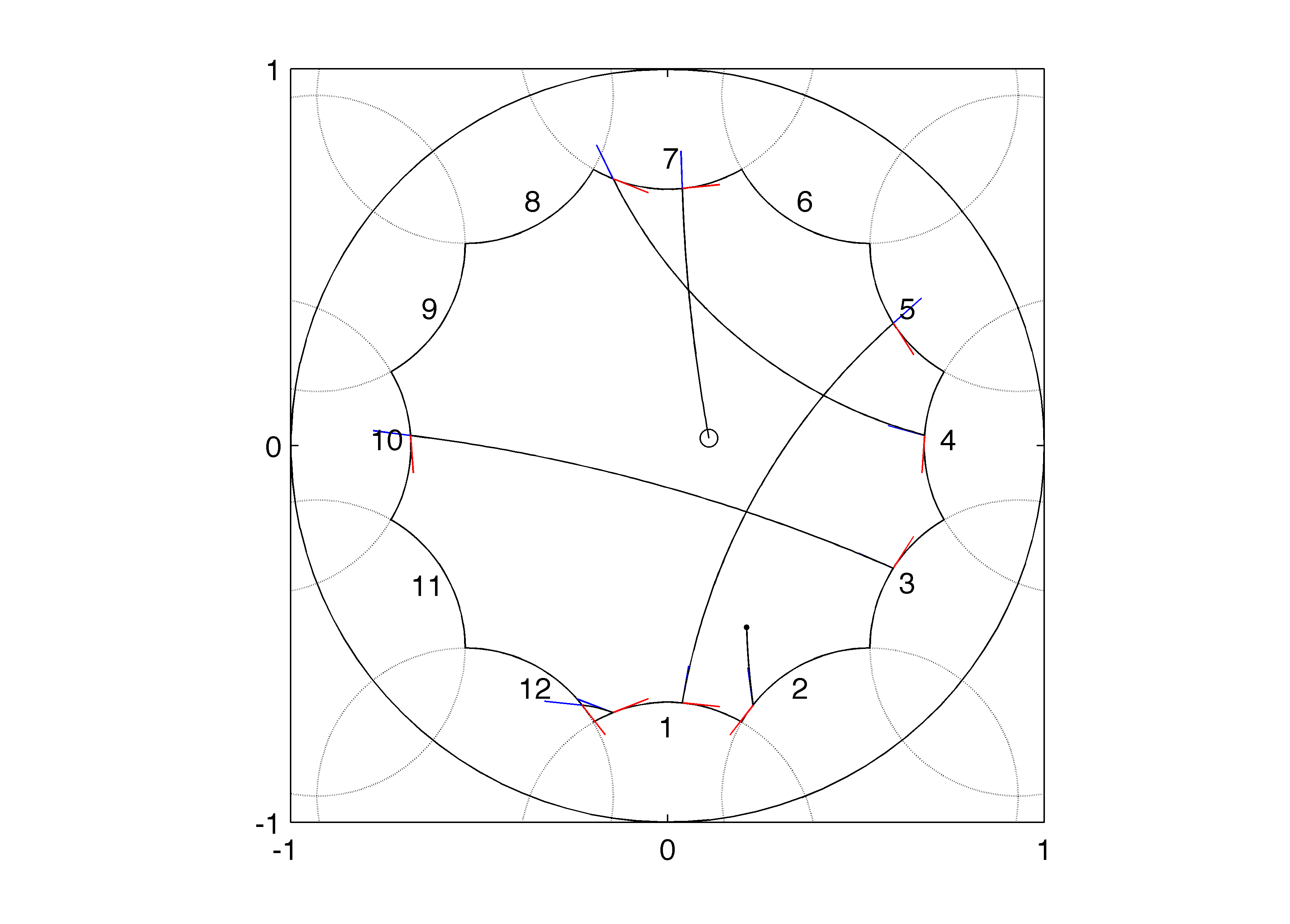}  
\caption{ \label{fig:genus2} (L) A surface with genus $g = 2$. For $g \ge 2$ there are surfaces of constant (or more generally global) negative curvature and whose geodesic flows are mixing (hence also ergodic). The construction of Adler and Flatto provides examples. For $g = 2$, the corresponding 12-gon can be recovered by cutting along the indicated paths. Incrementing $g$ by adding a handle also requires adding two more geodesic loops; one of the new loops and one of the old loops are separated into two arcs by their intersections with neighboring loops, and cutting along these arcs yields eight new edges. This construction underlies the pairing of edges of $F$, which is indicated explicitly here for $g=2$. (R) Model of the geodesic flow for genus $g = 2$ with 12-gon labels as indicated in the left panel. A sample trajectory is also indicated, with initial condition given by the open marker close to the center of the figure. Blue and red segments respectively indicate tangent directions to the flow and 12-gon at the latter's boundary.}
\end{figure}

We can construct a Poincar\'e or timing map and associated Markov partition for the geodesic flow. The first step is to instantiate edge pairing maps $T_j$ \emph{en route} to a map $T_R$ which will be the composition of the Poincar\'e map and an isometry. Then, we perform numerical calculations using the map $T_R$ and a Markov partition $\mathcal{R}$ for it. The advantage of this construction is that ``rectangles are rectangles,'' even though $T_R$ is nonlinear. For example, Figure \ref{fig:genus2partition} shows $\mathcal{R}$, $T_R \mathcal{R}$, and $T_R^2 \mathcal{R}$ for the genus $g=2$ case. From these Markov partitions, we obtain refinements $\mathcal{R}_m^\lor$ by intersecting rectangles in $T^0_R\mathcal{R}, \dots, T^m_R\mathcal{R}$.

\begin{figure}[htbp]
\centering
	\includegraphics[trim = 30mm 115mm 25mm 110mm, clip, width=\textwidth,keepaspectratio]{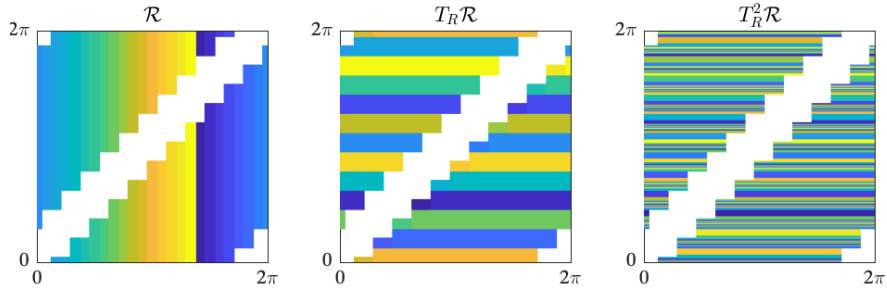}
\caption{ \label{fig:genus2partition} The Markov partitions $\mathcal{R}$, $T_R \mathcal{R}$, and $T_R^2 \mathcal{R}$ for the $g=2$ case. Note that rectangles become disconnected for $T_R^2 \mathcal{R}$.}
\end{figure}

Although as with hyperbolic toral automorphisms we have $\mu = \nu$ in this situation, this common measure is substantially more complicated than Lebesgue measure, viz.
\begin{equation}
\mu([x_1,x_2] \times [y_1,y_2]) = \int_{y_1}^{y_2} \int_{x_1}^{x_2} \frac{\lvert dx \ dy \rvert}{\lvert e^{ix}-e^{iy}\rvert^2}. \nonumber
\end{equation}
Consequently, the numerical computation of measures of rectangles in refinements of $\mathcal{R}$ is nontrivial. It turns out that $\beta/t_\infty$ diverges nearly exponentially for $\mathcal{R}_m^\lor$, whereas for the cat map $\beta/t_\infty$ provably converges due to linearity. 

However, for greedy refinements there is strong numerical evidence that $\lim \inf \beta$ is well-defined, nonzero, finite, and independent of genus $g$ (see Figures \ref{fig:geodesicbeta} and \ref{fig:geodesicbeta2}). This is due not only to the convergence of $\beta/t_\infty$, but also to the fact that the $L^2$ mixing time of the geodesic flow is $1/2$ for $g$ arbitrary.

\subsection{\label{sec:conclusion} Conclusion}

As archetypal Anosov systems, the cat map and the geodesic flow on a surface of constant negative curvature are deeply relevant to statistical physics. These and other Anosov systems exhibit Markovian symbolic dynamics that highlights both chaotic properties and correspondences with spin systems. The chaotic hypothesis seizes on these features to argue that Anosov systems are themselves archetypal model statistical-physical systems. For example, taking copies of the geodesic flow of \S \ref{sec:FreeParticle} with different initial conditions yields an ideal gas, and weakly coupling to this yields a thermometer.

As foreshadowed in \S \ref{sec:AnosovOverview}, the results of \cite{amaricci2007analyticity} and structural stability of Markov partitions for Anosov systems indicate that such a coupling leads to effects that vary analytically with the coupling strength. Moreover, the stability of rapid mixing \cite{field2007stability} gives at least one reason to expect that $t_\infty$ will also behave nicely as a function of coupling strength. SRB measures and mixing times for Anosov systems are well-behaved in many examples, and it is reasonable to expect this for physically relevant examples. 

While of course Markov structures (i.e., partitions or sections) are not unique, it is nevertheless evident that phenomena which hold for \emph{any} Markov structure on an Anosov system are likely to be of relevance to statistical physics. In this vein, the observed limiting behavior of $\beta/t_\infty$ as calculated on greedy partitions for both two-dimensional hyperbolic toral automorphisms and the geodesic flow is remarkable. While rigorously elaborating on this behavior would seem to require the development of new and nontrivial mathematics, it nevertheless appears to be generic.

Of particular importance are the implications that the evidence of this limiting behavior has for a proposed general theory of nonequilibrium steady states. 

\begin{svgraybox}
We have argued here for a comprehensive framework for nonequilibrium statistical physics that simultaneously incorporates and extends the formalism originally introduced by Ruelle and subsequently refined by Gallavotti, Cohen and others. The framework has as its goal a broad theory of nonequilibrium statistical physics that is truly intrinsic: i.e., that provides information about physical observables simply in terms of raw temporal information about the dynamics.
\end{svgraybox}

One reason to consider a proposal of the sort described here, in which the concept of (effective) temperature plays the central role, is because there is no generally accepted physical definition of entropy for non-equilibrium steady states. We hope that the ideas discussed here will serve to elicit fruitful investigations into the fundamental nature of stationary physical systems far from equilibrium.

\begin{acknowledgement}
I thank the editors for organizing the SPIGL Les Houches 2020 school as well as this volume. I am grateful to David Ford for initiating the development of the theory in the earlier parts of \S \ref{sec:2}; the applications that motivated it; and for much else besides. I thank colleagues at IDA and NPS for many fruitful discussions; and L for providing moral support. The ideas discussed here were originally developed with support from NSA, ARDA, DARPA, Equilibrium Networks, and BAE Systems FAST Labs. 
\end{acknowledgement}
%

\bibliographystyle{spphys}
\bibliography{HuntsmanSPIGLbib.bib}

\end{document}